\documentclass{article}

\PassOptionsToPackage{numbers, compress}{natbib}

\usepackage[final]{neurips_2024}

\usepackage{algorithm,algorithmic,amsfonts,amsmath,amsthm,amssymb,mathtools,mathrsfs,bbm,float}
\usepackage{graphicx,epic,eepic,epsfig,latexsym,verbatim,color}
\usepackage[caption=false]{subfig}
\usepackage{bm}
\usepackage{tikz}
\usetikzlibrary{chains}
\usetikzlibrary{fit}
\usetikzlibrary{arrows}		

\usepackage{epsfig}
\usetikzlibrary{shapes.symbols,patterns} 
\usepackage{pgfplots}

\usepackage[strict]{changepage}

\usepackage[marginal]{footmisc}
\usepackage{url}
\usepackage{thmtools, thm-restate}
\usepackage{derivative}
\usepackage{xfrac}
\usepackage{mleftright}

\newtheorem{proposition}{Proposition}
\newtheorem{lemma}[proposition]{Lemma}

\newtheorem{theorem}{Theorem}
\newtheorem{corollary}[proposition]{Corollary}

\newtheorem{remark}{Remark}

\def\squareforqed{\hbox{\rlap{$\sqcap$}$\sqcup$}}
\def\qed{\ifmmode\squareforqed\else{\unskip\nobreak\hfil
\penalty50\hskip1em\null\nobreak\hfil\squareforqed
\parfillskip=0pt\finalhyphendemerits=0\endgraf}\fi}
\def\endenv{\ifmmode\;\else{\unskip\nobreak\hfil
\penalty50\hskip1em\null\nobreak\hfil\;
\parfillskip=0pt\finalhyphendemerits=0\endgraf}\fi}

\newcounter{example}

\mathchardef\ordinarycolon\mathcode`\:
\mathcode`\:=\string"8000
\def\vcentcolon{\mathrel{\mathop\ordinarycolon}}
\begingroup \catcode`\:=\active
  \lowercase{\endgroup
  \let :\vcentcolon
  }

\usepackage{hyperref}
\hypersetup{colorlinks=true,citecolor=blue,linkcolor=blue,filecolor=blue,urlcolor=blue,breaklinks=true}
\usepackage{footnotebackref}
\usepackage[capitalize]{cleveref}
\usepackage{graphicx}
\usepackage{xcolor}
\usepackage{tablefootnote}
\usepackage{threeparttable}

\definecolor{darkblue}{RGB}{0,76,156}
\definecolor{darkkblue}{RGB}{0,0,153}
\definecolor{blue2}{RGB}{102,178,255}
\definecolor{darkred}{RGB}{195,0,0}

\RequirePackage[framemethod=default]{mdframed}
\newmdenv[skipabove=7pt,
skipbelow=7pt,
backgroundcolor=darkblue!15,
innerleftmargin=5pt,
innerrightmargin=5pt,
innertopmargin=5pt,
leftmargin=0cm,
rightmargin=0cm,
innerbottommargin=5pt,
linewidth=1pt]{tBox}

\newmdenv[skipabove=7pt,
skipbelow=7pt,
backgroundcolor=blue2!25,
innerleftmargin=5pt,
innerrightmargin=5pt,
innertopmargin=5pt,
leftmargin=0cm,
rightmargin=0cm,
innerbottommargin=5pt,
linewidth=1pt]{dBox}

\newmdenv[skipabove=7pt,
skipbelow=7pt,
backgroundcolor=darkred!15,
innerleftmargin=5pt,
innerrightmargin=5pt,
innertopmargin=5pt,
leftmargin=0cm,
rightmargin=0cm,
innerbottommargin=5pt,
linewidth=1pt]{rBox}

\newcommand{\nc}{\newcommand}
\nc{\ketbra}[2]{\lvert#1\rangle\!\langle#2\rvert}

\DeclarePairedDelimiter{\norm}{\lVert}{\rVert}
\DeclarePairedDelimiter{\abs}{\lvert}{\rvert}

\DeclarePairedDelimiter{\floor}{\lfloor}{\rfloor}

\DeclarePairedDelimiterX{\infdivx}[2]{(}{)}{%
  #1\;\delimsize\|\;#2%
}

\nc{\proj}[1]{| #1\rangle\!\langle #1 |}
\nc{\avg}[1]{\langle#1\rangle}
\nc{\smfrac}[2]{\mbox{$\frac{#1}{#2}$}}
\nc{\tr}{\operatorname{tr}}
\nc{\ox}{\otimes}
\nc{\dg}{\dagger}
\nc{\dn}{\downarrow}
\nc{\cA}{{\cal A}}
\nc{\cB}{{\cal B}}
\nc{\cC}{{\cal C}}
\nc{\cD}{{\cal D}}
\nc{\cE}{{\cal E}}
\nc{\cF}{{\cal F}}
\nc{\cG}{{\cal G}}
\nc{\cH}{{\cal H}}
\nc{\cI}{{\cal I}}
\nc{\cJ}{{\cal J}}
\nc{\cK}{{\cal K}}
\nc{\cL}{{\cal L}}
\nc{\cM}{{\cal M}}
\nc{\cN}{{\cal N}}
\nc{\cO}{{\cal O}}
\nc{\cP}{{\cal P}}
\nc{\cQ}{{\cal Q}}
\nc{\cR}{{\cal R}}
\nc{\cS}{{\cal S}}
\nc{\cT}{{\cal T}}
\nc{\cU}{{\cal U}}
\nc{\cV}{{\cal V}}
\nc{\cX}{{\cal X}}
\nc{\cY}{{\cal Y}}
\nc{\cZ}{{\cal Z}}
\nc{\cW}{{\cal W}}
\nc{\csupp}{{\operatorname{csupp}}}
\nc{\qsupp}{{\operatorname{qsupp}}}
\nc{\var}{{\operatorname{var}}}
\nc{\rar}{\rightarrow}
\nc{\lrar}{\longrightarrow}
\nc{\polylog}{{\operatorname{polylog}}}
\nc{\wt}{{\operatorname{wt}}}
\nc{\supp}{{\operatorname{supp}}}

\nc{\argmin}{{\operatorname{argmin}}}

\makeatletter
\newcommand{\tpmod}[1]{{\@displayfalse\pmod{#1}}}
\makeatother

\def\a{\alpha}

\def\n{\nu}
\def\x{\xi}

\nc{\RR}{{{\mathbb R}}}
\nc{\CC}{{{\mathbb C}}}
\nc{\FF}{{{\mathbb F}}}
\nc{\NN}{{{\mathbb N}}}
\nc{\ZZ}{{{\mathbb Z}}}
\nc{\PP}{{{\mathbb P}}}
\nc{\QQ}{{{\mathbb Q}}}
\nc{\UU}{{{\mathbb U}}}
\nc{\EE}{{{\mathbb E}}}
\nc{\id}{{\operatorname{id}}}

\nc{\CHSH}{{\operatorname{CHSH}}}

\nc{\rU}{\mbox{U}}

\nc{\ob}[1]{#1}

\nc{\SEP}{{\text{\rm SEP}}}
\nc{\NS}{{\text{\rm NS}}}
\nc{\LOCC}{{\text{\rm LOCC}}}
\nc{\PPT}{{\text{\rm PPT}}}
\nc{\EXT}{{\text{\rm EXT}}}
\nc{\Sym}{{\operatorname{Sym}}}

\nc{\ERLO}{{E_{\text{r,LO}}}}
\nc{\ERLOCC}{{E_{\text{r,LOCC}}}}
\nc{\ERPPT}{{E_{\text{r,PPT}}}}
\nc{\ERLOCCinfty}{{E^{\infty}_{\text{r,LOCC}}}}
\nc{\Aram}{{\operatorname{\sf A}}}

\newcommand{\eps}{\varepsilon}

\makeatletter
\def\grd@save@target#1{%
  \def\grd@target{#1}}
\def\grd@save@start#1{%
  \def\grd@start{#1}}
\tikzset{
  grid with coordinates/.style={
    to path={%
      \pgfextra{%
        \edef\grd@@target{(\tikztotarget)}%
        \tikz@scan@one@point\grd@save@target\grd@@target\relax
        \edef\grd@@start{(\tikztostart)}%
        \tikz@scan@one@point\grd@save@start\grd@@start\relax
        \draw[minor help lines,magenta] (\tikztostart) grid (\tikztotarget);
        \draw[major help lines] (\tikztostart) grid (\tikztotarget);
        \grd@start
        \pgfmathsetmacro{\grd@xa}{\the\pgf@x/1cm}
        \pgfmathsetmacro{\grd@ya}{\the\pgf@y/1cm}
        \grd@target
        \pgfmathsetmacro{\grd@xb}{\the\pgf@x/1cm}
        \pgfmathsetmacro{\grd@yb}{\the\pgf@y/1cm}
        \pgfmathsetmacro{\grd@xc}{\grd@xa + \pgfkeysvalueof{/tikz/grid with coordinates/major step}}
        \pgfmathsetmacro{\grd@yc}{\grd@ya + \pgfkeysvalueof{/tikz/grid with coordinates/major step}}
        \foreach \x in {\grd@xa,\grd@xc,...,\grd@xb}
        \node[anchor=north] at (\x,\grd@ya) {\pgfmathprintnumber{\x}};
        \foreach \y in {\grd@ya,\grd@yc,...,\grd@yb}
        \node[anchor=east] at (\grd@xa,\y) {\pgfmathprintnumber{\y}};
      }
    }
  },
  minor help lines/.style={
    help lines,
    step=\pgfkeysvalueof{/tikz/grid with coordinates/minor step}
  },
  major help lines/.style={
    help lines,
    line width=\pgfkeysvalueof{/tikz/grid with coordinates/major line width},
    step=\pgfkeysvalueof{/tikz/grid with coordinates/major step}
  },
  grid with coordinates/.cd,
  minor step/.initial=.2,
  major step/.initial=1,
  major line width/.initial=2pt,
}
\makeatother

\usepackage{thmtools}
\usepackage{thm-restate}
\usepackage{etoolbox}
\makeatletter
\def\problem@s{}
\newcounter{problems@cnt}

\newcommand{\allproblems}{\problem@s}
\makeatother

\usepackage{amsmath,amsfonts,amssymb,graphicx,mathtools,bm,tcolorbox,relsize}

\usepackage[utf8]{inputenc}
\usepackage[T1]{fontenc}
\usepackage[qm]{qcircuit}
\pgfplotsset{compat=1.18}
\usepackage{braket}
\usepackage{multirow}
\usepackage{booktabs}
\usepackage{tabularx}
\usepackage{array}
\usepackage[normalem]{ulem}
\usepackage{siunitx}

\DeclareMathOperator{\Poly}{Poly}

\DeclareMathOperator{\sgn}{sgn}
\DeclareMathOperator{\stp}{stp}
\newcommand{\relu}{\text{ReLU}}
\newcommand{\CNOT}{\text{CNOT}}
\newcommand{\holder}{\text{H{\"o}lder}}

\title{Non-asymptotic Approximation Error Bounds of Parameterized Quantum Circuits}

\author{
Zhan Yu\textsuperscript{\rm 1, 2}\thanks{Z.Y. and Q.C. contributed equally to this work} ~ Qiuhao Chen\textsuperscript{\rm 1}\footnotemark[1] ~ Yuling Jiao\textsuperscript{\rm 1, 3} ~ Yinan Li\textsuperscript{\rm 1, 3}\thanks{Yinan.Li@whu.edu.cn} ~ Xiliang Lu\textsuperscript{\rm 1, 3} \\ \textbf{Xin Wang\textsuperscript{\rm 4} ~ Jerry Zhijian Yang\textsuperscript{\rm 1, 3}}\thanks{zjyang.math@whu.edu.cn} \vspace{1em}\\
\footnotesize{ \textsuperscript{\rm 1} School of Mathematics and Statistics, Wuhan University, Wuhan 430072, China} \\
\footnotesize \textsuperscript{\rm 2} Centre for Quantum Technologies, National University of Singapore, 117543, Singapore \\
\footnotesize \textsuperscript{\rm 3} Hubei Key Laboratory of Computational Science, Wuhan 430072, China \\
\footnotesize \textsuperscript{\rm 4} Thrust of Artificial Intelligence, Information Hub, \\ \footnotesize Hong Kong University of Science and Technology (Guangzhou), Guangzhou 511453, China
}

\begin{document}

\maketitle

\begin{abstract}
Parameterized quantum circuits (PQCs) have emerged as a promising approach for quantum neural networks. However, understanding their expressive power in accomplishing machine learning tasks remains a crucial question. This paper investigates the expressivity of PQCs for approximating general multivariate function classes. Unlike previous Universal Approximation Theorems for PQCs, which are either nonconstructive or rely on parameterized classical data processing, we explicitly construct data re-uploading PQCs for approximating multivariate polynomials and smooth functions. We establish the first non-asymptotic approximation error bounds for these functions in terms of the number of qubits, quantum circuit depth, and number of trainable parameters. Notably, we demonstrate that for approximating functions that satisfy specific smoothness criteria, the quantum circuit size and number of trainable parameters of our proposed PQCs can be smaller than those of deep ReLU neural networks. We further validate the approximation capability of PQCs through numerical experiments. Our results provide a theoretical foundation for designing practical PQCs and quantum neural networks for machine learning tasks that can be implemented on near-term quantum devices, paving the way for the advancement of quantum machine learning.
\end{abstract}

\section{Introduction}
In quantum computing, one key area is to investigate if quantum computers could accelerate classical machine learning tasks in data analysis and artificial intelligence, giving rise to an interdisciplinary field known as \emph{quantum machine learning}~\cite{biamonte2017quantum}. As the quantum analogs of classical neural networks, \emph{parameterized quantum circuits} (PQCs)~\cite{benedetti2019parameterized} have gained significant attention as a prominent paradigm to yield quantum advantages. PQCs offer a concrete and practical way to implement quantum machine learning algorithms in noisy and intermediate-scale quantum (NISQ) devices~\cite{Preskill2018quantumcomputingin}, rendering them well-suited for a diverse array of tasks~\cite{kandala2017Hardware,cerezo2022VariationalQuantumState,cao2019QuantumChemistryAge,panDeepQuantumNeural2023,renExperimentalQuantumAdversarial2022,Huang2021Experimental,mitarai2018quantum,perez-salinas2020data}.

To establish the practical significance of quantum machine learning, an ongoing pursuit is to demonstrate their superiority in solving real-world learning problems compared to classical learning models, including the most commonly used deep neural networks~\cite{lecun2015DeepLearning}. Typical supervised learning tasks, such as image classification and price prediction, aim to construct a model to learn a mapping function from the input to output via training data sets. Essentially, the goal is to approximate multivariate functions. This viewpoint leads to the celebrated \emph{Universal Approximation Theorem}~\cite{cybenko1989approximation,hornik1989multilayer}, which limits what neural networks can theoretically learn. Recently, powerful tools from approximation theory have been utilized to establish a fruitful mathematical framework for understanding the ``black magic'' of deep learning by establishing non-asymptotic approximation error bounds of deep neural networks in terms of the \emph{width}, \emph{depth}, \emph{number of weights (neurons)} and function complexities, see e.g.\ Refs.~\cite{barron1993universal,yarotsky2017ErrorBoundsApproximations,yarotsky2018optimal,petersen2018optimal,yarotsky2020phase,shen2020deep,lu2021deep,shen2022optimal,weinan2022barron,jiao2023DNNReLUSineExponential,jiao2023deep} and references therein.

Substantial investigations have showcased the power of quantum machine learning for specific learning tasks~\cite{havlicek2019supervised,du2020expressive, liu2021RigorousRobustQuantum,Huang2021Information,jerbi2021parametrized,huang2021power, jerbi2023quantum, jager2023universal}. A fundamental question is whether the \emph{expressivity} of quantum machine learning models is as powerful as, or is more powerful than, the expressivity of classical machine learning models. This can be illustrated by proving universal approximation theorems for PQCs~\cite{schuld2021effect,gilvidal2020input,perez-salinas2021one,yu2022power,manzano2023parametrized,goto2021universal,gonon2023universal,qi2023TheoreticalErrorPerformance}, indicating that there exist PQCs with suitable parameter configurations to approximate target functions up to a given approximation accuracy. This will justify the power of PQCs to solve supervised learning tasks in a mathematical way. To further investigate whether PQCs are more expressive than the classical models or not, it is natural to examine the PQC approximation performance by establishing approximation error bounds for important function classes. Such quantitative error bounds are less known in the quantum setting, because the hypothesis functions generated by PQCs are more complicated than those generated by classical neural networks.

The difficulties of analyzing the PQC approximation performances can be partially overcome by allowing \emph{parameterized classical data processing}. Namely, trainable parameters are allowed not only in the quantum gates in PQCs but also in the classical data pre- and post-processing. This allows one to prove approximation error bounds following classical strategies~\cite{goto2021universal,qi2023TheoreticalErrorPerformance,gonon2023universal}. For instance,~\citet{goto2021universal} proved PQC approximation error rate for Lipschitz continuous functions in terms of the number of qubits and trainable parameters by incorporating trainable parameters in the measurement post-processing phase; similar results can also be obtained by utilizing Tensor-Train Network~\cite{qi2023TheoreticalErrorPerformance} or by linear transformations to preprocess the classical data. 

However, utilizing parameterized classical data processing makes it hard to distinguish whether the expressive power of PQCs comes from the classical or quantum parts. In fact, parameterized classical data processing enables one to directly convert the hypothesis functions generated by the quantum models into hypothesis functions generated by classical ones and adapt expressivity results for classical machine learning models to extract the expressivity of such quantum models. As a consequence, the resulting PQCs have very simple structures and short depth. It remains unknown whether one can prove approximation error bounds for PQCs without parameterized classical data processing. On the other hand, \citet{zhao2023learning} proved exponential lower bounds on the number of trainable parameters (in terms of the number of variables) needed for approximating bounded Lipschitz continuous functions using PQCs without parameterized classical data processing, illustrating that using PQCs to approximate Lipschitz functions still suffers from the \emph{curse of dimensionality} (CoD) met by classical deep neural networks~\cite{Grohsdeeplearning}. However, this does not rule out the possibility that one can achieve the same approximation rate with PQCs of \emph{smaller size} compared to classical deep neural networks.

In this paper, we explicitly construct \emph{the first} PQCs \emph{without} parameterized classical data processing for approximating multivariate polynomials and smooth functions; a glance at these constructed PQCs is illustrated in Fig.~\ref{fig:flowchart}. This eliminates the ambiguity regarding whether the expressivity originates from classical or quantum parts. We also establish \emph{non-asymptotic PQC approximation error bounds}, in the sense that the PQC approximation performances are characterized in terms of the number of qubits (width), the \emph{depth of PQCs}, the number of trainable parameters/gates (parameter count), and the function complexities. These results enable us to compare the approximation power of PQCs with that of classical neural networks. 
Notably, we show that for multivariate smooth functions, the quantum circuit size and the number of trainable parameters of our proposed PQCs demonstrate an improvement over the prior result of deep $\relu$ neural networks~\cite{lu2021deep}, one of the most commonly used neural network family in classical deep learning theory. Our proposed PQCs not only possess the universal approximation property but also achieve parameter efficiency comparable to classical neural networks, potentially leading to more efficient and scalable quantum machine learning algorithms for real-world tasks. 

\begin{figure}[ht!]
\centering
\includegraphics[width=0.99\textwidth]{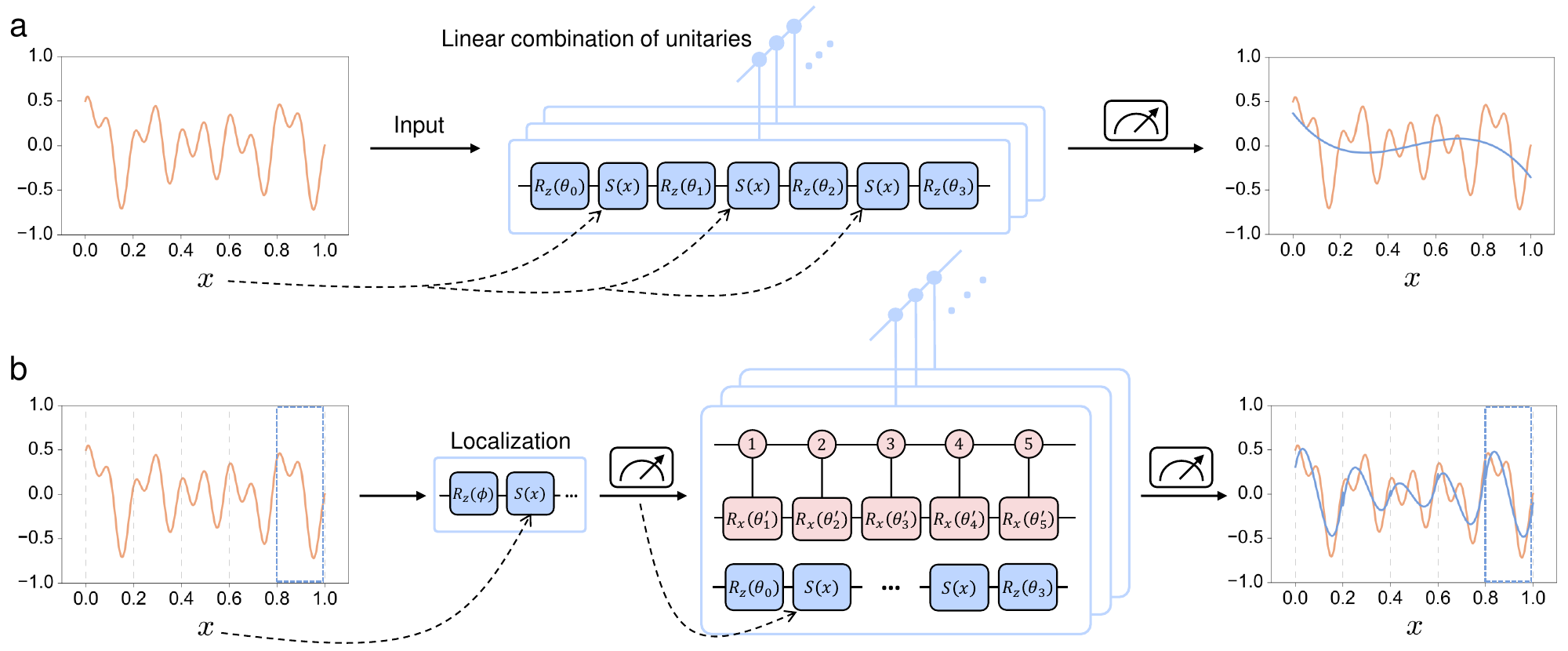}
\caption{\small{\textbf{Overview of PQCs for approximating continuous functions.} (a) Flowchart illustrating the strategy for using PQCs to approximate continuous functions via implementing Bernstein polynomials. The input data $x$ is encoded into the PQC through $S(x)$, with the PQC (blue background) capable of representing parity-constrained polynomials up to degree $3$ (as $x$ is encoded three times). The technique of linear combination of unitaries (LCU) is used to aggregate these polynomials together. The output of PQC derives from measurement with a specific observable. Fine-tuning trainable parameters in $R_Z$ gates yields a polynomial output depicted in the right panel.
(b) Flowchart illustrating the strategy of approximation via local Taylor expansions. We first apply a PQC to localize the input domain into $K=5$ regions. For example, for input $x\in [0.8, 1]$, PQC outputs $x^{\prime}=0.8$ as a fixed point. Then $x-x^{\prime}$ will be fed into a new PQC for implementing the local Taylor expansions at the fixed point $x'$, forming a nesting architecture. Control gates with pink backgrounds implement the Taylor coefficients. Fine-tuning trainable parameters in $R_X$ and $R_Z$ gates yields a piecewise polynomial with degree $3$ that approximates the target function.
}}
\label{fig:flowchart}
\end{figure}

\section{Preliminaries}
\paragraph{Quantum states.} 
The basic unit of information in quantum computing is the \emph{qubit}, which can exist in a superposition of the states 0 and 1 simultaneously, unlike classical bits that are restricted to either 0 or 1. A pure quantum state in the $d$-dimensional Hilbert space $\mathbb{C}^d$ is represented by the \emph{Dirac} notation $\ket{\phi}$. The conjugate transpose of $\ket{\phi}$ is denoted by $\bra{\phi}$. The inner product of two quantum states $\ket{\phi}$ and $\ket{\psi}$ is written as $\braket{\phi| \psi}$. An important property is that $\braket{\phi| \phi}=1$ for any pure state $\ket{\psi}$. By convention, the computational basis states for single-qubit systems are written as $\ket{0}=[1, 0]^T$ and $\ket{1}=[0, 1]^T$, where the superscript $T$ denotes the transpose. For $n$-qubit systems, the computational basis states are expressed as $\ket{j}\in\{\ket{0}, \ket{1}\}^{\otimes n}$, where $\otimes$ denotes the tensor product operation. 

\paragraph{Quantum gates.}
Quantum gates are building blocks of quantum circuits operating on quantum states. Unlike classical gates, quantum gates are reversible and described as unitary matrices. In quantum machine learning, common parameterized quantum gates include single-qubit Pauli rotation gates $R_X(\theta) = e^{-\theta X/2}$, $R_Y(\theta) = e^{-\theta Y/2}$, and $R_Z(\theta) = e^{-\theta X/2}$ that rotate a quantum state through angle $\theta$ around the corresponding axis, where the three Pauli operators are defined as:
\[
X = \begin{bmatrix}
  0 & 1 \\ 1 & 0
\end{bmatrix}, \quad
Y = \begin{bmatrix}
	0 & -i \\ i & 0
\end{bmatrix}, \quad
Z = \begin{bmatrix}
	1 & 0 \\ 0 & -1
\end{bmatrix},
\]
where $i$ represents the imaginary unit. Commonly used two-qubit quantum gates include CNOT gate that flips the target qubit if and only if the the control qubit is in $\ket{1}$.

\paragraph{Quantum measurement}
The quantum measurement is a procedure manipulating a quantum system to extract classical information. The simplest measurement is the computational basis measurement: For a single-qubit system $\ket{\psi}=\alpha\ket{0}+\beta\ket{1}$, the outcome is either $\ket{0}$ with probability $|\alpha|^2$ or $\ket{1}$ with probability $|\beta|^2$. These measurements project the quantum state onto the measured basis, collapsing the state itself. Observables, represented by Hermitian operators, correspond to measurable quantities in a quantum system like energy or position. Each observable has a set of possible outcomes (eigenvalues) and corresponding states (eigenvectors). When a measurement of an observable is performed, the outcome is one of the eigenvalues, and the state of the system collapses to the corresponding eigenvector. If we are measuring a state $\ket{\psi}$ using observable $\cO$, the expected value of outcome is $\bra{\psi}\cO\ket{\psi}$. This represents the average result one would expect from repeated measurements on identically prepared systems.
A comprehensive introduction to the fundamental notations and concepts of quantum computation can be found in~\cite{nielsen2010quantum}.

\paragraph{Data re-uploading PQCs.}
The PQCs we shall construct in this paper are of \emph{data re-uploading} type~\cite{perez-salinas2020data}, i.e., consisting of interleaved data encoding circuit blocks and trainable circuit blocks.
More precisely, let $\bm x$ be the input data vector and $\bm \theta = (\bm{\theta}_0, \ldots, \bm{\theta}_L)$ be a set of trainable parameter vectors. $S(\bm x)$ is a quantum circuit that encode $\bm x$ and $V(\bm \theta_j)$ is a trainable quantum circuit with trainable parameter vector $\bm\theta_j$. An $L$-layer data re-uploading PQC can be then expressed as
\begin{equation}\label{eqn:data reuploading PQC}
    U_{\bm\theta}(\bm{x}) = V(\bm{\theta_0}) \prod_{j = 1}^L S(\bm{x}) V(\bm{\theta_j}), 
\end{equation}
Applying $U_{\bm\theta}(\bm x)$ to an initial quantum state and measuring the output states provides a way to express functions on $\bm x$:
\begin{equation}\label{eqn:output_PQC}
    f_{U_{\bm\theta}}(\bm{x}) \coloneqq \bra{0} U^\dagger_{\bm\theta}(\bm{x}) \cO U_{\bm\theta}(\bm{x}) \ket{0}, 
\end{equation}

where $\cO$ is some Hermitian observable. The \emph{approximation capability} of the PQC $U_{\bm\theta}(\bm{x})$ can be characterized by the classes of functions that $f_{U_{\bm\theta}}(\bm{x})$ can approximate by tuning the trainable parameter vector $\bm\theta$. We then turn to an example of single-qubit PQCs approximating univariate functions. For the input $x \in [-1, 1]$, we utilized the Pauli $X$ basis encoding scheme~\cite{mitarai2018quantum} and defined the data encoding operator as a Pauli X rotation $S(x) \coloneqq e^{i \arccos(x) X}$. 
Interleaving the data encoding unitary $S(x)$ with some parameterized Pauli $Z$ rotations $R_Z(\theta)$ gives the circuit of data re-uploading PQC for one variable as $U_{\bm{\theta}}(x) \coloneqq R_Z(\theta_0) \prod_{j=1}^L S(x) R_Z(\theta_j)$
where $\bm \theta\ = (\theta_0, \ldots, \theta_L) \in \RR^{L+1}$ is a set of trainable parameters. Utilizing results from quantum signal processing~\cite{low2016methodology, low2017optimal, gilyen2019quantum}, there exists $\bm\theta\in \RR^{L+1}$ such that $U_{\bm{\theta}}(x)$ implements polynomial transformations $p(x) \in \RR[x]$ as $p(x) = \braket{+| U_{\bm{\theta}}(x)|+}$ for any $x\in[-1, 1]$
if and only if the degree of $p(x)$ is at most $L$, the parity of $p(x)$ is  $L \bmod 2$ \footnote{{A polynomial $p(x)$ has parity $0$ if all coefficients corresponding to odd powers of $x$ are $0$, and similarly $p(x)$ has parity $1$ if all coefficients corresponding to even powers of $x$ are $0$.}}, and $\abs{p(x)} \leq 1$ for all $x \in [-1, 1]$. Then, univariate functions that could be approximated by the specified polynomial $p(x)$ could also be approximated by the PQC $U_{\bm{\theta}}(x)$. Other than the real polynomials, there are also types of single-qubit PQC with Pauli $Z$ basis encoding that could implement complex trigonometric polynomials~\cite{yu2022power}.

\section{Expressivity of PQCs for multivariate continuous functions}
\subsection{Explicit construction of PQCs for multivariate polynomials}
Although PQCs for approximate univariate functions have been constructed and analyzed, they have not yet been generally extended to the case of multivariate functions. Current proofs of universal approximation for multivariate functions are nonconstructive~\cite{schuld2021effect, manzano2023parametrized} and require arbitrary circuit width, arbitrary multi-qubit global parameterized unitaries, and arbitrary observables. \citet{goto2021universal} proposed several constructions for approximating multivariate functions with the assistance of parameterized data pre-processing and post-processing, yielding a quantum-enhanced hybrid scheme rather than a purely quantum setting.

We now move to our explicit construction of PQCs for multivariate polynomials. A multivariate polynomial with $d$ variables and degree $s$ is defined as $p(\bm{x}) \coloneqq \sum_{\norm{\bm\a}_1 \leq s} c_{\bm\a} \bm{x^\a}$ where $\bm{x^\a} = x_1^{\a_1}x_2^{\a_2}\cdots x_d^{\a_d}$. To implement the multivariate polynomial $p(\bm{x})$, we first build a PQC to express a monomial $c_{\bm\a}\bm{x^\a}$. The construction is a trivial extension of the univariate case: We simply apply the single-qubit PQC with Pauli $X$ basis encoding on each $x_j$ to implement $x_j^{\a_j}$ for $1 \leq j \leq d$, respectively. The coefficient $c_{\bm\a} \in \RR$ could be implemented by any of these PQCs. Thus we could construct a PQC $U^{\bm\a}(\bm{x}) \coloneqq \bigotimes_{j=1}^d U_{\bm\theta_j}(x_j)$ such that $\bra{+}^{\otimes d}\! U^{\bm\a}(\bm{x}) \!\ket{+}^{\otimes d} = c_{\bm\a}\bm{x^\a}$. The depth of the PQC $U^{\bm\a}(\bm{x})$ is at most $2s+1$, the width is at most $d$, and the number of parameters is at most $s+d$.

Having PQCs that implement monomials, the next step is to aggregate monomials to implement the multivariate polynomial. A natural idea is to sum the monomial PQCs together as $U_p(\bm{x}) = \sum_{\norm{\bm\a}_1 \leq s} U^{\bm\a}(\bm{x})$. However, the addition operation in quantum computing is non-trivial as the sum of unitary operators is not necessarily unitary. To overcome this issue, we utilize \emph{linear combination of unitaries} (LCU)~\cite{childs2012hamiltonian} to implement the operator $U_p(\bm{x})$ on a quantum computer. Realizing the linear combination of PQCs $U^{\bm\a}(\bm{x})$ requires applying multi-qubit control on each $U^{\bm\a}(\bm{x})$, which could be further decomposed into linear-depth quantum circuits of CNOT gates and single-qubit rotation gates without using any ancilla qubit~\cite{dasilva2022lineardepth}. Then we can obtain the polynomial $p(\bm{x}) =\bra{+}^{\otimes d}\!  U_p(\bm{x}) \!\ket{+}^{\otimes d}$ by applying the Hadamard test on the LCU circuit. Summarizing the above, we establish the following theorem about using PQCs to implement multivariate polynomials. A formal description of such PQCs is given in \cref{Appendix:B}.
\begin{theorem}\label{prop:pqc_poly}
    For any multivariate polynomial $p(\bm{x})$ with $d$ variables and degree $s$ such that $\abs{p(\bm{x})} \leq 1$ for $\bm{x} \in [0,1]^d$, there exists a PQC $W_p(\bm{x})$ such that
    \begin{equation}
        f_{W_p}(\bm{x}) \coloneqq \bra{0} W^\dagger_{p}(\bm{x}) Z^{(0)} W_{p}(\bm{x}) \ket{0} =  p(\bm{x}) 
    \end{equation}
    where $Z^{(0)}$ is the Pauli $Z$ observable on the first qubit. The width of the PQC is $O(d+\log s + s\log d)$, the depth is $O(s^2 d^s (\log s + s\log d))$, and the number of parameters is $O(sd^s(s+d))$.
\end{theorem}
Note that the initial state in the Hadamard test is $\ket{0}^{\otimes d}$ since $\ket{+}^{\otimes d}$ could be easily prepared by applying Hadamard gates on $\ket{0}^{\otimes d}$. Measuring the first qubit of $W_p(\bm{x})$ for $O(\frac{1}{\eps^2})$ times is needed to estimate the value of $p(\bm{x})$ up to an additive error $\eps$. We could further use the amplitude estimation algorithm~\cite{brassard2002quantum} to reduce the overhead while increasing the circuit depth by $O(\frac{1}{\eps})$.

\subsection{PQC approximation for continuous functions}
Polynomials play a central role in approximation theory.
The celebrated Weierstrass approximation theorem (see e.g.~\cite[Sec.\ 10.2.2]{Davidson2002Realanalysiswithrealapplications}) indicates that polynomials are sufficient to approximate continuous univariate functions. 
For multivariate functions, their approximation can be implemented using Bernstein polynomials~\cite{heitzinger2002simulation, foupouagnigni2020multivariate}. We shall apply these results to prove PQC approximation error bounds for multivariate Lipschitz continuous functions.

For a $d$-variable continuous function $f: [0,1]^d \to \RR$, the multivariate Bernstein polynomial with degree $n\in\NN^{+}$ of $f$ is defined as
\begin{equation}\label{eqn:Bernstein_poly}
    B_n(\bm{x}) \coloneqq \sum_{k_1=0}^n \cdots \sum_{k_d=0}^n f\bigl(\frac{\bm{k}}{n}\bigr) \prod_{j=1}^d \binom{n}{k_j} x_j^{k_j}(1-x_j)^{n-k_j},
\end{equation}
where $\bm{k} = (k_1, \ldots, k_d) \in \{0, \ldots, n\}^d$. It is known that Bernstein polynomials converge uniformly to $f$ on $[0,1]^d$ as $n\to\infty$ \cite{heitzinger2002simulation, foupouagnigni2020multivariate}. The PQC constructed in \cref{prop:pqc_poly} could implement the Bernstein polynomial with proper rescaling, which implies that the PQC is a universal approximator for any bounded continuous functions.

\begin{theorem}[The Universal Approximation Theorem of PQC]\label{thm:pqc_uat}
    For any continuous function $f:[0,1]^d \to [-1,1]$, given an $\eps > 0$, there exist an $n\in \NN$ and a PQC $W_b(\bm{x})$ with width $O(d\log n)$, depth $O(dn^d\log n)$ and the number of trainable parameters $O(dn^d)$ such that
    \begin{equation}
        \abs{f(\bm{x}) - f_{W_b}(\bm{x})} \leq \eps
    \end{equation}
    for all $\bm{x}\in [0,1]^d$, where $f_{W_b}(\bm{x}) \coloneqq \bra{0} W^\dagger_{b}(\bm{x}) Z^{(0)} W_{b}(\bm{x}) \ket{0}$.
\end{theorem}
\cref{thm:pqc_uat} serves as the quantum counterpart to the universal approximation theorem of classical neural networks. Moreover, the PQCs that universally approximate continuous functions are explicitly constructed without any impractical assumption, improving the previous results presented in Refs.~\cite{schuld2021effect, manzano2023parametrized}. Moreover, for continuous functions $f$ satisfying the Lipschitz condition, $\abs{f(\bm{x}) - f(\bm{y})} \leq \ell \norm{\bm{x} - \bm{y}}_\infty$ for any $\bm{x},\bm{y}$, the approximation rate of Bernstein polynomials could be quantitatively characterized in terms of the degree $n$, the number of variables $d$ and the Lipschitz constant $\ell$~\cite{foupouagnigni2020multivariate}. Thus a non-asymptotic error bound for PQC approximating Lipschitz continuous functions could be obtained as follows. 
\begin{theorem}\label{thm:approx_Lipschitz}
Given a Lipschitz continuous function $f: [0, 1]^d \to [-1, 1]$ with a Lipschitz constant $\ell$, for any $\eps > 0$ and $n \in \NN$, there exists a PQC $W_b(\bm{x})$ with such that $f_{W_b}(\bm{x}) \coloneqq \bra{0} W^\dagger_{b}(\bm{x}) Z^{(0)} W_{b}(\bm{x}) \ket{0}$ satisfies
\begin{equation}
    \abs{f(\bm{x})-f_{W_b}(\bm{x})} \leq \eps + 2\biggl( \Bigl(1+\frac{\ell^2}{n\eps^2}\Bigr)^d-1 \biggr) \leq \eps + d2^{d}\frac{\ell^2}{n\eps^2}
\end{equation}
for all $\bm{x}\in[0, 1]^d$. The width of the PQC is $O(d\log n)$, the depth is $O\bigl(dn^{d}\log{n}\bigr)$, and the number of parameters is $O(dn^d)$.
\end{theorem}

We prove these theorems in \cref{Appendix:C}.
Although a quantitative approximation error bound is characterized in \cref{thm:approx_Lipschitz}, we could find that $n$ must be sufficiently large to obtain a good precision, yielding an extremely deep PQC. This inefficiency is essentially due to the intrinsic difficulty of using a single global polynomial to approximate a continuous function uniformly. A possible approach that may overcome the obstacle is to use local polynomials to achieve a piecewise approximation, which we will discover in the next section.

\subsection{PQC approximation for H{\"o}lder smooth functions}
To achieve a piecewise approximation of multivariate functions, we follow the path of classical deep neural networks approximation~\cite{petersen2018optimal, lu2021deep, jiao2023deep}, which utilizes multivariate Taylor series to approximate target functions in small local regions.

We focus on $\holder$ smooth functions. Let $\beta = s + r >0$, where $r \in (0, 1]$ and $s\in \NN^{+}$. For a finite constant $B_0 > 0$, the $\beta$-$\holder$ class of functions $\cH^\beta([0,1]^d, B_0)$ is defined as
\begin{equation}
    \cH^\beta([0,1]^d, B_0)\!=\!\Bigl\{ f\!:[0,1]^d\!\to\!\RR, \max_{\norm{\bm\a}_1\leq s} \norm{\partial^{\bm\a}f}_{\infty}\!\leq\!B_0, \max_{\norm{\bm\a}_1=s} \sup_{\bm{x} \neq \bm{y}} \frac{\abs{\partial^{\bm\a}f(\bm{x})-\partial^{\bm\a}f(\bm{y})}}{\norm{\bm{x}-\bm{y}}_2^r}\!\leq\! B_0\Bigr\},
\end{equation}
where $\partial^{\bm\a} = \partial^{\a_1} \cdots \partial^{\a_d}$ for $\bm\a=(\a_1,\ldots, \a_d)\in \NN^d$. We note that $\holder$ smooth functions are natural generalizations of various continuous functions: When $\beta\in(0,1)$, $f$ is $\holder$ continuous with order $\beta$ and $\holder$ constant $B_0$; when $\beta = 1$, $f$ is Lipschitz continuous with Lipschitz constant $B_0$; when $1<\beta\in\NN$, $f\in C^s([0,1]^d)$, the class of $s$-smooth functions whose $s$-th partial derivatives exist and are bounded. As shown in \citet{petersen2018optimal}, for any $\beta$-$\holder$ smooth function $f \in \cH^\beta([0,1]^d, B_0)$, its local Taylor expansion at some fixed point $\bm{x}_0 \in [0,1]^d$ satisfies
\begin{equation}\label{eqn:taylor_expansion}
    \abs[\Big]{f(\bm{x}) - \sum_{\norm{\bm{\a}}_1 \leq s} \frac{\partial^{\bm{\a}} f(\bm{x_0})}{\bm\a !} (\bm{x} - \bm{x_0})^{\bm{\a}}} \leq d^s \norm{\bm{x} - \bm{x_0}}^\beta_2
\end{equation}
for all $\bm{x} \in [0,1]^d$, where $\bm\a! = \a_1!\cdots \a_d!$. Next, we show how to construct PQCs to implement the Taylor expansion of $\beta$-$\holder$ functions in the following three steps.

\paragraph{Localization.} To utilize the $\holder$ smoothness, we need to first localize the entire region $[0, 1]^d$. 
The motivation of localization is to determine the local point $\bm{x_0}$ in \cref{eqn:taylor_expansion} so that the distance between $\bm{x}$ and $\bm{x_0}$ is fairly small. An intuitive configuration is illustrated in \cref{fig:fig2}, where the stars represent the local points. 
Given $K \in \NN$ and $\Delta \in (0, \frac{1}{3K})$, for each $\bm{\eta} = (\eta_1, \ldots, \eta_d) \in \{0, 1, \ldots, K-1\}^d$, we define
\begin{equation}
    Q_{\bm{\eta}} \coloneqq \Bigl\{ \bm{x}=(x_1,\ldots,x_d): x_i \in \bigl[ \frac{\eta_i}{K}, \frac{\eta_i+1}{K} - \Delta \cdot 1_{\eta_i < K - 1}\bigr] \Bigr\}.
\end{equation}
By the definition of $Q_{\bm{\eta}}$, the region $[0, 1]^d$ is approximately divided into small hypercubes $\bigcup_{\bm\eta}Q_{\bm{\eta}}$ and some trifling region $\Lambda(d, K, \Delta) \coloneqq [0,1]^d \setminus (\bigcup_{\bm\eta}Q_{\bm{\eta}})$, as illustrated in \cref{fig:fig2}. 
\begin{figure}[ht!]
\centering
\subfloat{\includegraphics[width=0.57\textwidth]{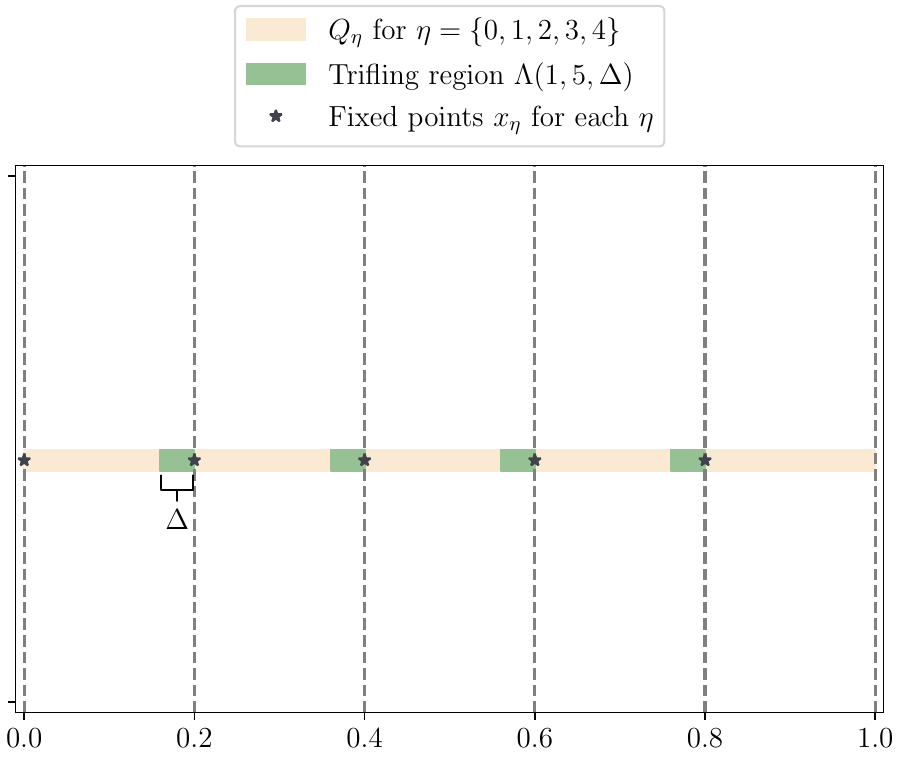}}
\hfill
\subfloat{\includegraphics[width=0.384\textwidth]{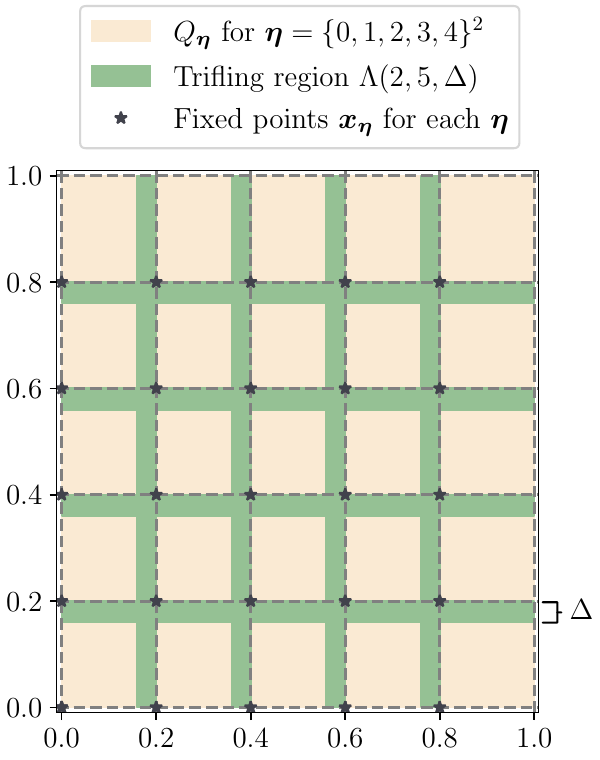}}
\caption{\small{\textbf{An illustration of localization.} The left panel demonstrates the localization $\bigcup_{\bm{\eta}} Q_{\bm{\eta}}$ for $K=5$ and $d=1$. The right panel shows the case of localization for $K=5$ and $d=2$. The ``volume'' of the trifling region $\Lambda(d, K, \Delta)$ is no more than $dK\Delta$. 
}}
\label{fig:fig2}
\end{figure}

We construct a PQC that maps all $\bm{x} \in Q_{\bm{\eta}}$ to some fixed point $\bm{x_\eta} = \frac{\bm\eta}{K}$ in $Q_{\bm{\eta}}$, i.e., approximating the piecewise-constant function $D(\bm{x}) = \frac{\bm{\eta}}{K}$ if $\bm{x} \in Q_{\bm{\eta}}$.
We describe our construction for $d=1$, where $D(x) = \frac{k}{K}$ if $x \in [\frac{k}{K}, \frac{k+1}{K} - \Delta \cdot 1_{k < K - 1}]$ for $k = 0, \ldots, K-1$. The multivariate case could be naturally generalized by applying $D(x)$ to each variable $x_j$. The idea is to construct a polynomial that approximates the function $D(x)$ based on the polynomial approximation to the sign function~\cite{low2017quantum}, which a single-qubit PQC can then implement. Generalizing to the multivariate localization, there exists a PQC $W_{D}(\bm{x})$ of depth $O(\frac{1}{\Delta}\log\frac{K}{\eps})$ and width $O(d)$ such that the output $f_{W_D}(\bm{x})$ maps $\bm{x}$ to the corresponding fixed point $\bm{x_\eta}$ with precision $\eps$.
We can obtain an estimation of $\bm{\eta}$ using $\lfloor Kf_{W_D}(\bm{x}) \rfloor$.

\paragraph{Implementing the Taylor coefficients.} 
Next, we use PQC to implement the Taylor coefficients $\xi_{{\bm{\eta}}, \bm\a} \coloneqq \frac{\partial^{\bm\a} f(\bm{x_\eta})}{\bm\a !} \in [-1, 1]$ for each ${\bm{\eta}} = (\eta_1, \ldots, \eta_d) \in \{0, 1, \ldots, K-1\}^d$ and $\bm\a$, which is essentially a point-fitting problem. Then we could construct a PQC $U_{co}^{\bm\a} = \sum_{{\bm{\eta}}} \ketbra{{\bm{\eta}}}{{\bm{\eta}}} \otimes R_X(\theta_{\bm{\eta},\bm \a})$
such that $\bra{{\bm{\eta}}, 0} U_{co}^{\bm\a} \ket{{\bm{\eta}}, 0} = \xi_{{\bm{\eta}}, \bm \alpha}$, where $\ket{{\bm{\eta}}} = \ket{\eta_1} \otimes \cdots \otimes \ket{\eta_d}$ and $\theta_{\bm \eta,\bm \a}=2\arccos(\xi_{\bm\eta,\bm\a})$. The depth of $U_{\bm{\a}}$ is $O(K^d)$, the width is $O(d\log K)$, and the number of parameters is $O(K^d)$. Note that the state $\ket{{\bm{\eta}}}$ can be prepared using basis encoding on the provided $\bm\eta$ $=\lfloor Kf_{W_D}(\bm{x}) \rfloor$ from the localization step.

\paragraph{Implementing multivariate Taylor series.}
To implement the multivariate Taylor expansion of a function at some fixed point $\bm{x_\eta}$, we first build a PQC to represent a single term in the Taylor series, which could be done by combining the PQC, which implements the Taylor coefficients and the PQC which implements monomials, i.e., constructing $U^{\bm\a}_{\bm\eta}(\bm{x}) \coloneqq U_{co}^{\bm{\a}} \otimes U^{\bm\a}(\bm{x} - \bm{x_\eta})$. The depth of $U^{\bm\a}_{\bm\eta}(\bm{x})$ is $O(K^d + s)$, the width is $O(d\log K)$, and the number of parameters is at most $K^d+s+d$. The next step is to aggregate single Taylor terms together to implement the truncated Taylor expansion of the target function. We use LCU to construct the PQC $U_t(\bm{x}, \bm{x_\eta}) \coloneqq \sum_{\norm{\bm\a}_1 \leq s} U^{\bm\a}_{\bm\eta}(\bm{x})$ so that we can implement the Taylor expansion of the function $f$ at point $\bm{x_\eta}$ as $\bra{\bm\eta, 0}\!\bra{+}^{\otimes d} U_t(\bm{x}, \bm{x_\eta}) \ket{\bm\eta, 0}\!\ket{+}^{\otimes d}$.

We construct a nested PQC as $U_t(\bm{x}, f_{W_D}(\bm{x}))$, such that for any input $\bm{x}$, the corresponding fixed point could be determined by the localization PQC. Such a PQC could be used, together with the Hadamard test, to approximate $\holder$ smooth functions. In particular, we prove the approximation error bound of our constructed PQC based on the error rate of Taylor expansion in \cref{eqn:taylor_expansion}.

\begin{theorem}\label{thm:approx_holder}
    Given a function $f \in \cH^\beta([0,1]^d, 1)$ with $\beta = r + s$, $r\in(0,1]$ and $s\in \NN^{+}$, for any $K\in\NN$ and $\Delta\in(0, \frac{1}{3K})$, there exists a PQC $W_{t}(\bm{x})$ such that $f_{W_t}(\bm{x}) \coloneqq \bra{0} W^\dagger_{t}(\bm{x}) Z^{(0)} W_{t}(\bm{x}) \ket{0}$ satisfies
    \begin{equation}
        \abs{f(\bm{x}) - f_{W_t}(\bm{x})} \leq d^{s+\beta/2}K^{-\beta}
    \end{equation}
    for $\bm{x} \in \bigcup_{{\bm{\eta}}} Q_{\bm{\eta}}$. The width of the PQC is $O(d\log K + \log s + s\log d)$, the depth is $O(s^2 d^sK^d(\log s + s\log d + d\log K)) + \frac{1}{\Delta}\log K)$, and the number of parameters is $O(sd^s(s+d+K^d) + \frac{d}{\Delta} \log K)$.
\end{theorem}
The proof can be found in \cref{Appendix:D}.
Note that the PQC in \cref{thm:approx_holder} consists of two nested parts and its depth is counted as the sum of two PQCs for simplicity. We have established the uniform convergence property of PQCs for approximating $\holder$ smooth function on $[0, 1]^d$ except for the trifling region $\Lambda(d, K, \Delta)$. The Lebesgue measure of such a trifling region is no more than $dK\Delta$. We can set $\Delta=K^{-d}$ with no influence on the size of the constructed PQC, and a similar approximation error bound in the entire region $[0,1]^d$ under the $L^2$ distance could be obtained.

\section{Numerical experiments}
	This section presents numerical experiments to illustrate the expressivity of our proposed PQCs in approximating multivariate functions. We focus on approximating a bivariate polynomial function \[f(x, y)=\frac{(x^2+y-1.5\pi)^2+(x+y^2+\pi)^2+(x+y-0.5\pi)^2}{5\pi^2},\] over the domain $(x, y)\in[0, 1]^2$. The approximation process involves two separate steps: (1) Learning a piecewise-constant function, $D(x)=\frac{k}{K}$ if $x\in[\frac{k}{K}, \frac{k+1}{K})$, using a single-qubit PQC, where $K\in\mathbb{N}^{+}$ determines the number of intervals for the piecewise-constant function. (2) Learning the Taylor expansion of $f(x, y)$ using multi-qubit PQCs based on~\cref{thm:approx_holder}. Both learning processes are implemented on a Gold 6248 2.50 GHz Intel(R) Xeon(R) CPU.
	
	We randomly sample \num{200} data points within the domain $[0,1]$ to create training and test datasets for $D(x)$. A single-qubit PQC with adjustable parameters $L=764$ ($L=996$) is used to learn $D(x)$ with $K=2$ ($K=10$). Each parameter of the PQC is randomly initialized within the range \([0, \pi]\). We use the Adam optimizer~\cite{kingma2015adam} with a learning rate of \num{0.01} to minimize the Mean Squared Error (MSE) loss function during training. The training process was limited to a maximum of \num{300} iterations with a batch size of 100 data points. Early termination occurred if the MSE reached below $10^{-4}$. The achieved MSE on the test data was \num{3.57d-4} ($K=2$) and \num{1.04d-4} ($K=10$). The numerical results are visualized in~\cref{fig:fig3}.
	
	\begin{figure}[ht!]
		\centering
		\includegraphics[width=0.35\textwidth]{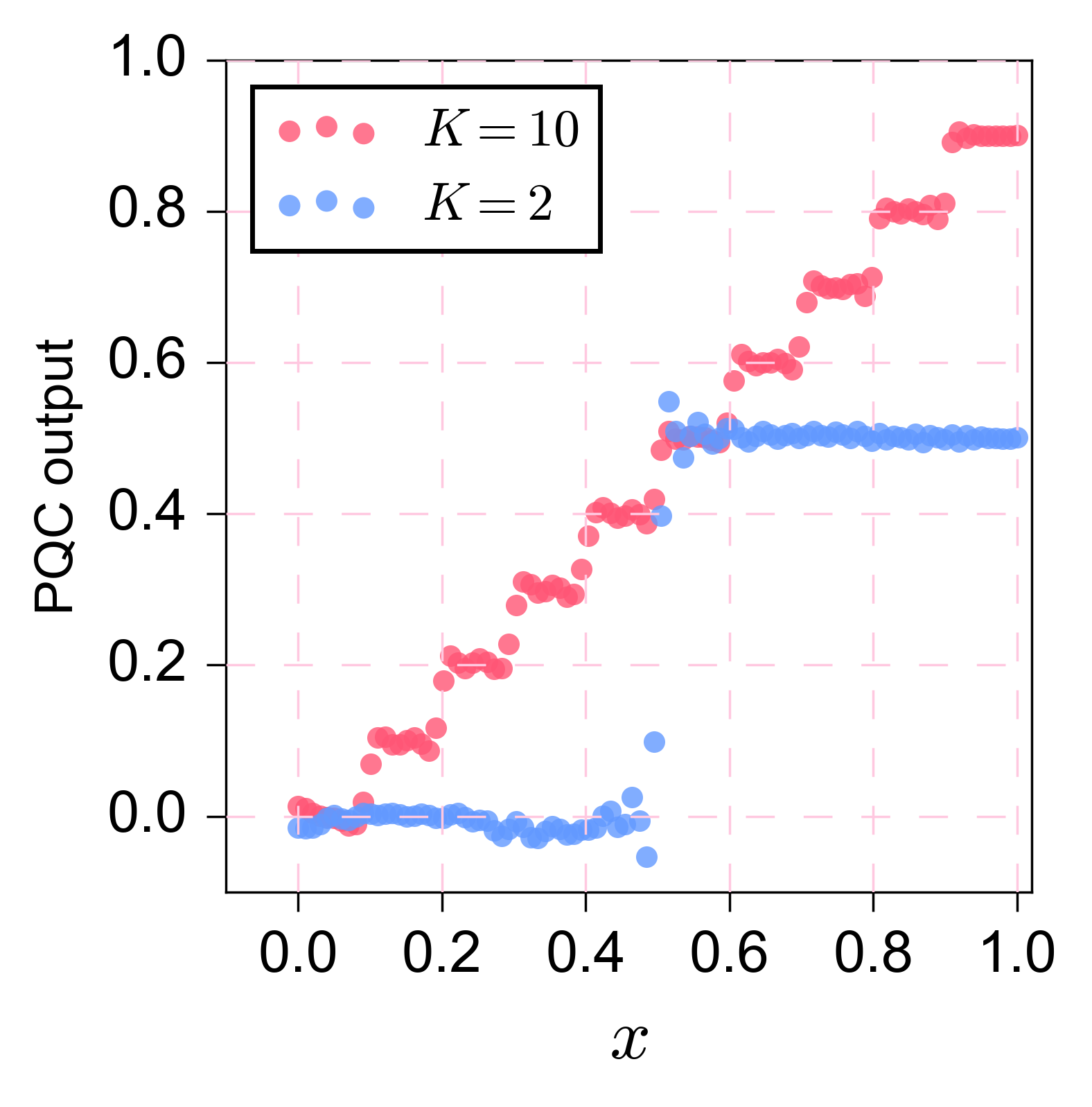}
		\caption{\small{\textbf{Simulation results of localization.} We use single-qubit PQCs to approximate the localization function $D(x)$ for $K=2$ and $K=10$ respectively.
		}}
		\label{fig:fig3}
	\end{figure}
	
	Similar to the previous step, we randomly sampled \num{200} data points within the domain $[0, 1]^2$ to create training and test datasets for $f(x, y)$. A nested PQC structure was designed. It combined \num{12} two-qubit PQCs with a depth of \num{2}, allowing the approximation of a degree-4 polynomial through a combination of lower-degree ones. Additionally, Taylor coefficients were stored in a separate matrix of size $K^2\times 12$. The number of trainable parameters varied from \num{120} ($K=2$) to \num{1272} ($K=10$), each initialized randomly from \([0, \pi]\). The Adam optimizer with a learning rate of 0.01 was used to minimize the MSE loss during training. The training was limited to \num{500} iterations with a batch size of 100, with early termination for MSE below \num{d-4}. The achieved MSE on the test data was \num{2.22d-4} ($K=2$) and \num{9.82d-5} ($K=10$). \cref{fig:fig4} visualizes the results. As $K$ increases, the PQC demonstrates improved approximation performance, aligning with the theoretical findings.
	
	\begin{figure}[ht!]
		\centering
		\includegraphics[width=0.99\textwidth]{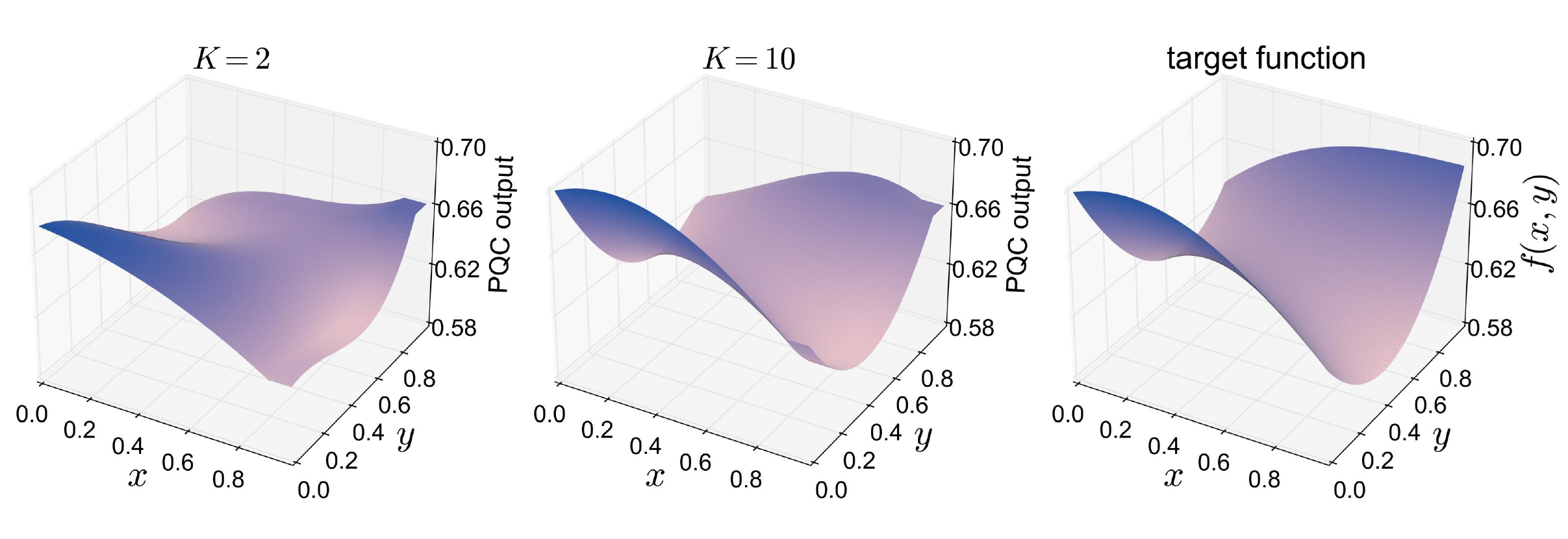}
		\caption{\small{\textbf{Simulation results for learning $f(x, y)$.} The left two panels are derived by interpolating and smoothing the output values of PQC on 100 test data points.
		}}
		\label{fig:fig4}
	\end{figure}

\section{Discussion}

To the best of our knowledge, our results establish the first explicit PQC constructions for approximating Lipschitz continuous and $\holder$ smooth functions with quantitative approximation error bounds.
These results open up the possibility of comparing the size of PQCs and the size of classical deep neural networks for accomplishing the same function approximation tasks and see if there is any quantum advantage in terms of the model size and the number of trainable parameters.
Here, we mainly focus on the comparison with the results of approximation errors of classical machine learning models. In classical deep learning, the deep feed-forward neural network (FNN) equipped with the rectified linear unit (ReLU) activation function is one of the most commonly used models. The quantitative approximation error bounds of ReLU FNNs for approximating continuous functions have been recently established, including the nearly optimal approximation error bounds of ReLU FNNs for smooth functions~\cite{lu2021deep}. We briefly compare the approximation errors of PQCs and ReLU FNNs in terms of width, depth and the number of trainable parameters. Detailed comparisons can be found in \cref{Appendix:E}.

We consider multivariate smooth functions in $C^s_u([0,1]^d)$ (the unit ball of $C^s([0,1]^d)$) with smooth index $s \in \NN$ as the target functions in our comparison. Note that smooth functions with smooth index $s$ are exactly $(s+1)$-$\holder$ smooth functions by definition. For simplicity, we first show the case of $s=2$. To achieve the same approximation error $\eps$ (say some constant), we need to set $K_{Q}=\Theta(d^2/\sqrt{\eps})$ for the constructed PQCs from \cref{thm:approx_holder} and set $K_C=\Theta(2^{d/2}/\sqrt{\eps})$ for the constructed near-optimal ReLU FNNs from Ref.~\cite{lu2021deep}. Substituting the choices of $K$'s in the sizes of PQCs and ReLU FNNs, we have
\begin{equation}
    \frac{\text{Width of PQC}\times \text{Depth of PQC}}{\text{Width of FNN}\times \text{Depth of FNN}}=O\Bigl(\frac{d^3K_Q^{d}}{2^{d+3}K_C^{d/2}}\Bigr)=O\Bigl(\frac{\eps^{-d/4}}{2^{d^{2}-d\log d }}\Bigr).
\end{equation}
One can obtain a similar relation for the number of required parameters in PQCs and ReLU FNNs for approximating smooth functions and extend these results to any $2\leq s<d$, which holds relevance in numerous real-world applications (e.g., the input dimension $d$ is \num{784} for the MNIST dataset and is \num{150528} for the ImageNet dataset~\cite{deng2009imagenet}, and empirically $s\leq 10$).  Therefore, to achieve the same approximation error, the required quantum circuit size and number of parameters of PQCs is exponentially smaller than the required network size and number of parameters of ReLU FNNs proposed in Ref.~\cite{lu2021deep}.

Aiming to understand and continuously expand the range of problems that can be addressed using quantum machine learning, we have demonstrated the approximation capabilities of PQC models in supervised learning. We characterized the approximation error of PQCs in terms of the model size, delivering a deeper understanding of the expressive power of PQCs that is beyond the universal approximation properties. With these results, we can unlock the full potential of these models and drive advancements in quantum machine learning. Notably, by comparing our results with the near-optimal approximation error bound of classical $\relu$ neural networks, we demonstrate an improvement over the classical models on approximating high-dimensional functions that satisfy specific smoothness criteria, quantified by an improvement on the model size and the number of parameters.

Unlike many other investigations in the universal approximation properties of PQC models~\cite{havlicek2019supervised,du2020expressive, liu2021RigorousRobustQuantum,Huang2021Information,jerbi2021parametrized,huang2021power, jerbi2023quantum, jager2023universal}, our constructions of PQCs for approximating broad classes of continuous functions do not rely on any impractical assumptions. All the variables take the form of parameters within single-qubit rotation gates, avoiding any classical parameterized pre-processing or post-processing. Ultimately, our research provides valuable insights into the theoretical underpinnings of PQCs in quantum machine learning and paves the way for leveraging its capabilities in machine learning for both classical and quantum applications. 

In this work, we introduce a novel nested PQC structure, which significantly improves the approximation capabilities. Future work could focus on exploring more powerful PQC constructions based on our proposed idea and understanding the capabilities and limitations of PQCs in more practical tasks even with real-world data. Developing efficient training strategies for PQCs, such as accelerated methods that achieve faster convergence rates, will also be interesting.

\begin{ack}
Part of this work was done when Z.Y. was visiting Wuhan University. Z.Y. thanks Patrick Rebentrost for helpful discussions. The authors thank the helpful comments from the anonymous reviewers. This work is supported by the National Key Research and Development Program of China (No.\ 2020YFA0714200), the National Nature Science Foundation of China (No.\ 62302346, No.\ 12125103, No.\ 12071362, No.\  12371424, No.\ 12371441) and supported by the ``Fundamental Research Funds for the Central Universities''.
\end{ack}

\bibliographystyle{unsrtnat}
\bibliography{ref}

\clearpage

\begin{center}
\textbf{
{\Large{Supplementary Material}}}
\end{center}

\appendix
\renewcommand{\theequation}{\thesection.\arabic{equation}}
\renewcommand{\thetheorem}{S\arabic{theorem}}
\renewcommand{\theproposition}{S\arabic{proposition}}
\renewcommand{\thelemma}{S\arabic{lemma}}
\renewcommand{\theremark}{S\arabic{remark}}
\renewcommand{\thecorollary}{S\arabic{corollary}}
\renewcommand{\thefigure}{S\arabic{figure}}
\renewcommand{\thetable}{S\arabic{table}}
\setcounter{equation}{0}
\setcounter{table}{0}
\setcounter{proposition}{0}
\setcounter{lemma}{0}
\setcounter{corollary}{0}
\setcounter{figure}{0}
\setcounter{remark}{0}
\setcounter{figure}{0}
\setcounter{table}{0}

\tableofcontents
\medskip

\section{Preliminaries}
In this section, we will first present some essential mathematical foundations for deriving the main results of this work. Moreover, to contextualize our work within the existing literature, we comprehensively review relevant studies in \cref{SM:related work}.

\subsection{Notation}
We unify the notations throughout the whole work. The univariate polynomial ring over a field $\mathbb{F}$ is symbolized as $\mathbb{F}[x]$, with the variable $x$ representing the input. The ring of Laurent polynomial $\mathbb{F}[x, x^{-1}]$ is an extension of the polynomial ring obtained by adding inverses of $x$. The collection of natural numbers is represented by the symbol $\mathbb{N}:=\{1, 2, 3, \dots\}$, while the set of non-negative integers is denoted as $\mathbb{N}_0:=\{0\}\cup\mathbb{N}$. The $1$-norm of a vector $\bm{\alpha}=(\alpha_1, \alpha_2, \dots, \alpha_d)$ is denoted by $\|\bm{\alpha}\|_1:=|\alpha_1|+|\alpha_2|+\cdots +|\alpha_d|$. 

\subsection{Data re-uploading PQCs}
In this section, we review the concept of data re-uploading PQC and define the PQC we use in this paper. The data re-uploading PQC is a quantum circuit that consists of interleaved data encoding circuit blocks and trainable circuit blocks~\cite{gilvidal2020input, perez-salinas2020data}. More precisely, let $\bm x$ be the input data vector and $\bm \theta = (\bm{\theta_0}, \ldots, \bm{\theta_L})$ be a set of trainable parameters. $S(\bm x)$ is a quantum circuit that encode $\bm x$ and $V(\bm \theta_j)$ is a trainable quantum circuit with trainable parameter vector $\bm\theta_j$. An $L$-layer data re-uploading PQC can be then expressed as  
\begin{equation}\label{eqnS:data reuploading PQC}
    U_{\bm\theta}(\bm{x}) = V(\bm{\theta_0}) \prod_{j = 1}^L S(\bm{x}) V(\bm{\theta_j}),
\end{equation}
Applying $U_{\bm\theta}(\bm x)$ to a quantum state and measuring the output states provides a way to express functions on $\bm x$. The \emph{expressivity} of the data re-uploading PQC model can be characterized by the classes of functions that it can implement. It is common to build data encoding circuits and trainable circuits using the most prevalent Pauli rotation operators,
\begin{equation}
    R_X(\theta) = \begin{bmatrix}
        \cos\frac{\theta}{2} & -i\sin\frac{\theta}{2}\\[1ex]
        -i\sin\frac{\theta}{2} & \cos\frac{\theta}{2}
    \end{bmatrix},\quad
    R_Y(\theta) = \begin{bmatrix}
        \cos\frac{\theta}{2} & -\sin\frac{\theta}{2}\\[1ex]
        \sin\frac{\theta}{2} & \cos\frac{\theta}{2}
    \end{bmatrix},\quad
    R_Z(\theta) = \begin{bmatrix}
        e^{-i \frac{\theta}{2}} & 0\\[1ex]
        0 & e^{i \frac{\theta}{2}}
    \end{bmatrix}.
\end{equation}
Different data encoding schemes lead to different types of data re-uploading PQCs. 

In some cases, trainable parameters are also included both during the initial data encoding phase and the final processing of measurement outcomes. These PQCs are considered to have \emph{hybrid} structures. For instance, in the models proposed by Refs.~\cite{gilvidal2020input,perez-salinas2021one,gonon2023universal}, each input data is multiplied by a specific trainable parameter and subsequently subjected to $R_Z$ gates during the data encoding stage. In a similar vein, Refs.~\cite{goto2021universal, gonon2023universal} incorporate trainable weights into each measurement outcome generated by the constructed PQCs, aggregating these weighted outcomes to produce the final result. Such a structure makes it hard to judge whether the expressive power comes from the classical or quantum part.

\subsubsection{Implementing real polynomials}

We first introduce the data re-uploading PQC for implementing real univariate polynomials. We utilize the so-called \emph{Pauli $X$ basis encoding}~\cite{mitarai2018quantum}: The data encoding unitary is a single-qubit rotation defined as
\begin{equation}\label{eqnS:data encoding S}
    S(x) \coloneqq e^{i \arccos(x) X}=\begin{pmatrix}
        x & i\sqrt{1-x^2}\\
        i\sqrt{1-x^2} & x
    \end{pmatrix},
\end{equation}
where $x\in[-1,1]$ is the input data. Then interlaying the data encoding unitary $S(x)$ with some parameterized Pauli $Z$ rotations $R_Z(\theta)$ gives the circuit of data re-uploading PQC for one variable as
\begin{equation}\label{eqnS:PQC_x}
    U_{\bm{\theta}}(x) \coloneqq R_Z(\theta_0) \prod_{j=1}^L S(x) R_Z(\theta_j),
\end{equation}
where $\bm \theta\ = (\theta_0, \ldots, \theta_L) \in \RR^{L+1}$ is a set of trainable parameters. The PQC in~\cref{eqnS:PQC_x} can be used to implement polynomial transformations on input $x$, as shown in the following lemma.

\begin{lemma}[\cite{gilyen2019quantum}]\label{lem:qsp}
There exists $\bm \theta \in \RR^{L+1}$ such that
\begin{equation}
    U_{\bm{\theta}}(x) = \begin{pmatrix}
        P(x) & iQ(x)\sqrt{1-x^2}\\
        iQ^*(x)\sqrt{1-x^2} & P^*(x)
    \end{pmatrix}
\end{equation}
if and only if polynomials $P, Q \in \CC[x]$ satisfy
\begin{enumerate}
    \item $\deg(P)\leq L$ and $\deg(Q) \leq L-1$,
    \item $P$ has parity $L \bmod 2$ and $Q$ has parity $(L-1) \bmod 2$\footnote{{For a polynomial $P\in\CC[x]$, $P$ has parity $0$ if all coefficients corresponding to odd powers of $x$ are $0$, and similarly $P$ has parity $1$ if all coefficients corresponding to even powers of $x$ are $0$.}},
    \item $\forall x\in [-1, 1]$, $\abs{P(x)}^2 + (1-x^2)\abs{Q(x)}^2 = 1$.
\end{enumerate}
\end{lemma}
As shown in the above lemma, one could implement a polynomial transformation $\Poly(x)$ such that $\Poly(x) = \braket{0|U_{\bm{\theta}}(x) |0} = P(x)$. Notice that the achievable polynomial $\Poly(x)$ implemented in this way is limited to $P(x)$ for which there exists a polynomial $Q(x)$ satisfying the conditions of \cref{lem:qsp}. As the target polynomial is often real in practice, we could overcome such a limitation by defining $\Poly(x) = \braket{+| U_{\bm{\theta}}(x)|+} = \Re(P(x)) + i\Re(Q(x))\sqrt{1-x^2}$. Then we can achieve any real polynomials with parity $L \bmod 2$ such that $\deg(\Poly(x)) \leq L$, and $\abs{\Poly(x)}\leq 1$ for all $x\in[-1, 1]$.

\begin{corollary}[\cite{gilyen2019quantum}]\label{cor:qsp}
There exists $\bm\theta \in \RR^{L+1}$ such that
\begin{equation}
    p(x) = \braket{+| U_{\bm{\theta}}(x)|+}
\end{equation}
if and only if the real polynomial $p(x) \in \RR[x]$ satisfies
\begin{enumerate}
    \item $\deg(p(x))\leq L$,
    \item $p(x)$ has parity $L \bmod 2$ \footnote{{A polynomial $p(x)$ has parity $0$ if all coefficients corresponding to odd powers of $x$ are $0$, and similarly $p(x)$ has parity $1$ if all coefficients corresponding to even powers of $x$ are $0$.}},
    \item $\forall x\in [-1, 1]$, $\abs{p(x)}\leq 1$.
\end{enumerate}
\end{corollary}
\begin{remark}
    The results of PQC with Pauli $X$ basis encoding presented here have been established in the technique of quantum signal processing~\cite{low2016methodology, low2017optimal, gilyen2019quantum}, which uses interleaving signal operators and signal processing operators to transform the input signal. The QSP circuit could be identified as a PQC in the context of quantum machine learning.
\end{remark}

\subsubsection{Implementing trigonometric polynomials}
Other than the real polynomials, there are also types of single-qubit PQC with Pauli $Z$ basis encoding that could implement complex trigonometric polynomials~\cite{yu2022power}. The data encoding unitary is a single-qubit rotation in the Pauli $Z$ basis
\begin{equation}
	S(x) \coloneqq R_Z(x) = \begin{pmatrix}
		e^{ix/2} & 0\\
		0 & e^{-ix/2}
	\end{pmatrix},
\end{equation}
where $x \in \RR$ is the data. By interleaving the data encoding unitary $S(x)$ with trainable gates $R_Y(\theta)R_Z(\phi)$, the PQC is defined as
\begin{equation}\label{eqn:PQC_z}
	U_{\bm{\theta}, \bm{\phi}}(x) \coloneqq R_Z(\omega)R_Y(\theta_0)R_Z(\phi_0) \prod_{j = 1}^L S(x) R_Y(\theta_j)R_Z(\phi_j),
\end{equation}
where $\bm\theta = (\theta_0, \ldots, \theta_L) \in \RR^{L+1}$, $\bm\phi = (\phi_0, \ldots, \phi_L) \in \RR^{L+1}$ and $\omega \in \RR$. The following lemma characterizes the correspondence between PQC with $\sigma_z$ basis encoding and complex trigonometric polynomials.

\begin{lemma}[\cite{yu2022power}]\label{lem:trig_qsp}
There exist $\bm\theta,\bm\phi \in \RR^{L+1}$ and $\omega \in \RR$ such that
\begin{equation}\label{eqn:trig form of QSP_}
    U_{\bm{\theta}, \bm{\phi}}(x) = \begin{pmatrix}
        P(x) & -Q(x)\\
        Q^*(x) & P^*(x)
    \end{pmatrix}
\end{equation}
if and only if Laurent polynomials $P, Q \in \CC[e^{ix/2}, e^{-ix/2}]$ satisfy
\begin{enumerate}
    \item $\deg(P)\leq L$ and $\deg(Q) \leq L$,
    \item $P$ and $Q$ have parity $L \bmod 2$,
    \item $\forall x\in \RR$, $\abs{P(x)}^2 + \abs{Q(x)}^2 = 1$.
\end{enumerate}
\end{lemma}
Note that Laurent polynomials in $\CC[e^{ix/2}, e^{-ix/2}]$ with parity $0$ are Laurent polynomials in $\CC[e^{ix}, e^{-ix}]$ without parity constraints, which implies that the trigonometric QSP could implement complex trigonometric polynomials.

\begin{corollary}[\cite{yu2022power,wang2023quantum}]\label{corS:trig_qsp_proj}
	There exist $\bm\theta,\bm\phi \in \RR^{2L+1}$ and $\omega \in \RR$ such that
	\begin{equation}
		t(x) = \bra{0} U_{\bm{\theta}, \bm{\phi}}(x) \ket{0}
	\end{equation}
	if and only if the complex-valued trigonometric polynomial $t(x) = \sum_{j=-L}^L c_j e^{ijx}$ satisfies $\abs{t(x)} \leq 1$ for all $x \in \RR$.
\end{corollary}

\subsection{Related work in PQC approximation}\label{SM:related work}
In this subsection, we review prior literature related to the approximation capabilities of PQCs, which characterizes how the architectural properties of a PQC affect the resulting functions it can fit, and its ensuing performance. 
After a systematic comparison, we conclude that our results provide precise error bounds for continuous function approximation and make no assumptions about the constructed PQCs. More importantly, all the variables in our proposal take the form of parameters within rotation gates and remain distinct from the data encoding gates to avoid any classical computational influence, thus preserving the inherent quantum property of our approach. 

In theoretical machine learning, statistical complexity is a notion that measures the inherent richness characterizing a given hypothesis space. There are various statistical complexity measures, including the Vapnik-Chervonenkis (VC) dimension~\cite{vapnik1982necessary}, the metric entropy~\cite{tikhomirov1993eentropy}, the Gaussian complexity~\cite{bartlett2003rademacher}, and the Rademacher complexity~\cite{bartlett2003rademacher}, etc. To gauge the statistical complexity of PQCs, \citet{du2022efficient} have explored the covering entropy of PQCs in terms of the number of quantum gates and the measurement observable. \citet{bu2022statistical} have investigated the dependence of the Rademacher complexity of PQCs on the resources, width, depth, and the property of input and output registers. The assessment of PQCs has extended to encompass an array of statistical complexity measures, including the Pseudo-Dimension, as delineated in \citet{caro2020pseudodimension}, and the VC dimension, as expounded upon in \citet{chen2022general}. 
Furthermore, the evaluation of PQC expressivity has extended its purview to metrics rooted in information theory. \citet{abbas2021power} have evaluated PQC expressivity through the prism of the effective dimension, a data-dependent metric contingent upon the Fisher information. In a parallel endeavor, \citet{du2020expressive} have concentrated their attention on generative tasks, employing entanglement entropy as a metric for quantifying PQC expressivity. 
It is important to underscore that, while statistical complexity metrics and information-inspired metrics provide invaluable insights into the `volume' of hypothesis spaces, they do not precisely delineate the functions amenable to representation by these models. 

To further explore the intricacies of PQCs and their expressivity, an alternative avenue of research has emerged, as highlighted by recent studies~\cite{schuld2021effect,gilvidal2020input, yu2022power, perez-salinas2021one, manzano2023parametrized}. They rewrote the PQC output, i.e., the inner product between an input quantum state and a variational observable, in the form of partial Fourier series. This innovative perspective introduces a more nuanced toolbox for assessing PQC expressivity, offering fresh insights within the quantum machine learning domain, notably with respect to the universal approximation property (UAP).
However, it is imperative to underscore that many investigations employing Fourier expansion have been predicated upon certain impractical assumptions. These assumptions encompass the demand for arbitrary parameterized global unitaries and observables, thus posing significant challenges to the practical implementation of the constructed quantum circuits. The existence proof of universal approximation also does not explicitly give approximation error bounds of PQCs.

A very general approach to expressiveness in the context of approximation is the method of nonlinear widths by \citet{devore1989optimal} that concerns the approximation of a family of functions under the assumption of a continuous dependence of the model on the approximated function. \citet{perez-salinas2021one} have proved that single-qubit data re-uploading PQCs are universal function approximators, inheriting the famous universal approximation theorem for neural networks~\cite{cybenko1989approximation,hornik1991approximation}. In a quantum-enhanced context, \citet{goto2021universal} have constructed PQCs to approximate any continuous function guided by the Stone-Weierstrass theorem. 
\citet{qi2023TheoreticalErrorPerformance} have studied the approximation error of PQCs enhanced by tensor-train networks. Their investigation focused on smooth functions, considering factors such as the number of qubits and quantum measurement counts. Furthermore, \citet{gonon2023universal} have defined a specific hypothesis space consisting of non-oscillating functions, drawing inspiration from \citet{barron1993universal} and devised PQCs for approximating such functions without encountering the curse of dimensionality (CoD). Notably, the mitigation of CoD arises from their specific hypothesis space definition and is also observed within the domain of classical neural network~\cite{e2022barron}. It is essential to acknowledge that these works unveil a hybrid nature, blurring the boundaries between classical and quantum domains in circuit construction. The hybrid structure manifests in the data encoding phase and becomes evident in the weighted summation of outputs from foundational quantum circuits. Consequently, whether the powerful expressivity comes from the classical part or the quantum part of hybrid models is unclear.

In our present work, we make no assumptions in the construction of the PQCs. In our PQC model, all variables take the form of parameters within rotation gates. Besides, these trainable parameters remain distinct from the data encoding gates to avoid any classical computational influence. These properties ensure that our constructed PQCs retain practicality and remain firmly rooted within the quantum domain.

\section{Implementing multivariate polynomials using PQCs}\label{Appendix:B}
\subsection{Implementing multivariate real polynomials}\label{subsec:PQC_multi_poly}
A multivariate polynomial with $d$ variables and degree $s\in\NN$ is defined as
\begin{equation}
	p(\bm{x}) \coloneqq \sum_{\norm{\bm\a}_1 \leq s} c_{\bm\a} \bm{x^\a},
\end{equation}
where $\bm{x} = (x_1, \ldots, x_d) \in \RR^d$, $\bm\a = (\a_1, \ldots, \a_d) \in \NN^d$, $c_{\bm \alpha}\in\RR$ and $\bm{x^\a} = x_1^{\a_1}x_2^{\a_2}\cdots x_d^{\a_d}$. To implement the multivariate polynomial $p(\bm{x})$, we first build a PQC to express a monomial $c_{\bm\a}\bm{x^\a} = c_{\bm\a}x_1^{\a_1}x_2^{\a_2}\cdots x_d^{\a_d}$, where $\abs{c_{\bm\a}\bm{x^\a}} \leq 1$ for $\bm{x} \in[0,1]^d$ and $\norm{\bm{\a}}_1 \leq s$. We apply the single-qubit PQC with Pauli $X$ basis encoding defined in \cref{eqnS:PQC_x} on each $x_j$ for $1 \leq j \leq d$, respectively. 

\begin{lemma}\label{lem:monomial_PQC}
	Given a monomial $c_{\bm\a}\bm{x^\a} = c_{\bm\a}x_1^{\a_1}x_2^{\a_2}\cdots x_d^{\a_d}$ such that $\abs{c_{\bm\a}\bm{x^\a}} \leq 1$ for all $\bm{x} \in[0,1]^d$ and $\norm{\bm{\a}}_1 \leq s$ for $s\in\NN$, there exists a PQC $U^{\bm\a}(\bm{x})$ such that
	\begin{equation}
		\bra{+}^{\otimes d}\! U^{\bm\a}(\bm{x}) \!\ket{+}^{\otimes d} = c_{\bm\a}\bm{x^\a}.
	\end{equation}
	The width of the PQC is at most $d$, the depth is at most $2s+1$, and the number of parameters is at most $s+d$.
\end{lemma}

\begin{proof}
	By \cref{cor:qsp}, there exist $d$ single-qubit PQCs $U_{\bm\theta_1}^{\a_1}(x_1), U_{\bm\theta_2}^{\a_2}(x_2), \ldots, U_{\bm\theta_d}^{\a_d}(x_d)$ such that
	\begin{align*}
		\braket{+|U_{\bm\theta_1}^{\a_1}(x_1)|+} &= c_{\bm\a}x_1^{\a_1},\\
		\braket{+|U_{\bm\theta_2}^{\a_2}(x_2)|+} &= x_2^{\a_2},\\
		\cdots\\
		\braket{+|U_{\bm\theta_d}^{\a_d}(x_d)|+} &= x_d^{\a_d},
	\end{align*}
	where the number of layers of each PQC is $L_j = \a_j$ for $1 \leq j \leq d$. We then define a $d$-qubit PQC as
	\begin{equation}
		U^{\bm\a}(\bm{x}) = \bigotimes_{j=1}^d U_{\bm\theta_j}^{\a_j}(x_j),
	\end{equation}
	which gives
	\begin{equation}
		\bra{+}^{\otimes d}\! U^{\bm\a}(\bm{x}) \!\ket{+}^{\otimes d} = \prod_{j=1}^d \braket{+|U_{\bm\theta_j}^{\a_j}(x_j)|+} = c_{\bm\a}\bm{x^\a}.
	\end{equation}
	Since $\norm{\bm\a}_1 = \sum_{j=1}^d \alpha_j \leq s$, we can conclude that the depth of $U^{\bm\a}(\bm{x})$ is at most $2s+1$ and the number of parameters in $U^{\bm\a}(\bm{x})$ is at most $s + d$.
\end{proof}

The next step is to combine monomials together to implement the multivariate polynomial. Specifically, we would like to implement the following (unnormalized) operator
\begin{equation}
	U_p(\bm{x}) \coloneqq \sum_{\norm{\bm\a}_1 \leq s} U^{\bm\a}(\bm{x})
\end{equation}
so that we can implement an (unnormalized) polynomial as
\begin{equation}
	\bra{+}^{\otimes d}\!  U_p(\bm{x}) \!\ket{+}^{\otimes d} = \sum_{\norm{\bm\a}_1 \leq s} \bra{+}^{\otimes d}\!U^{\bm\a}(\bm{x})\!\ket{+}^{\otimes d} = \sum_{\norm{\bm\a}_1 \leq s} c_{\bm\a} \bm{x^\a} = p(\bm{x}).
\end{equation}
We denote $T$ the number of terms in the summation and observe that it can be bounded as
\begin{equation}
	T = \sum_{\norm{\bm\a}_1 \leq s} 1 = \sum_{j=0}^s \sum_{\norm{\bm\a}_1=j} 1 \leq \sum_{j=0}^s d^s \leq (s+1)d^s.
\end{equation}
For convenience, we rewrite the normalized target operator with $\bm\a$ being an indexed variable as
\begin{equation}\label{eqnS:sum_PQC}
	U_p(\bm{x}) = \sum_{j=1}^T \frac{1}{T} U^{\bm\a^{(j)}}(\bm{x}).
\end{equation}
However, the addition operation in quantum computing is non-trivial as the sum of unitary operators is not necessarily unitary. To sum the monomials together, we utilize the technique of \emph{linear combination of unitaries} (LCU)~\cite{childs2012hamiltonian} to implement the operator $U_p(\bm{x})$ in \cref{eqnS:sum_PQC} on a quantum computer. We first construct a unitary operator $F$ such that
\begin{equation}
	F\ket{0} = \frac{1}{\sqrt{T}} \sum_{j=1}^T \ket{j}.
\end{equation}
The unitary $F$ could be simply implemented by Hadamard gates. Next, we construct a controlled unitary
\begin{equation}
	U_c(\bm{x}) = \sum_{j=1}^T \ketbra{j}{j} \otimes U^{\bm\a^{(j)}}(\bm{x}).
\end{equation}
Note that each $\ketbra{j}{j} \otimes U^{\bm\a^{(j)}}(\bm{x})$ could be constructed using $(\log T)$-qubit controlled Pauli rotation gates, as $U^{\bm\a^{(j)}}(\bm{x})$ consisting of single-qubit Pauli rotation gates. The $(\log T)$-qubit controlled gates could be further decomposed into quantum circuits of CNOT gates and single-qubit rotation gates in $O(\log T)$ circuit depth without using any ancilla qubit. We refer to the detailed implementation of these multi-controlled gates to \citet{dasilva2022lineardepth}.
Then the unitary $W_{lcu} = (F^\dagger\otimes I) U_c (F \otimes I)$ satisfies that
\begin{equation}
	W_{lcu} \ket{0} \ket{+}^{\otimes d} = \ket{0} U_p(\bm{x}) \ket{+}^{\otimes d} + \ket{\perp},
\end{equation}
where $(\bra{0}\otimes I) \ket{\perp} = 0$.
Notice that
\begin{equation}
	\bra{0}\bra{+}^{\otimes d} W_{lcu} \ket{0}  \ket{+}^{\otimes d} = \bra{+}^{\otimes d}  U_p(\bm{x}) \ket{+}^{\otimes d} = p(\bm{x}).
\end{equation}
To obtain the polynomial $p(\bm{x})$, we could estimate $\bra{0}\bra{+}^{\otimes d}  W_{lcu} \ket{0}  \ket{+}^{\otimes d}$ using the Hadamard test.

\renewcommand\thetheorem{\ref{prop:pqc_poly}}
\setcounter{theorem}{\arabic{theorem}-1}
\begin{theorem}
    For any multivariate polynomial $p(\bm{x})$ with $d$ variables and degree $s$ such that $\abs{p(\bm{x})} \leq 1$ for $\bm{x} \in [0,1]^d$, there exists a PQC $W_p(\bm{x})$ such that
    \begin{equation}
        f_{W_p}(\bm{x}) \coloneqq \bra{0} W^\dagger_{p}(\bm{x}) Z^{(0)} W_{p}(\bm{x}) \ket{0} =  p(\bm{x}) 
    \end{equation}
    where $Z^{(0)}$ is the Pauli $Z$ observable on the first qubit. The width of the PQC is $O(d+\log s + s\log d)$, the depth is $O(s^2 d^s (\log s + s\log d))$, and the number of parameters is $O(sd^s(s+d))$.
\end{theorem}
\renewcommand{\thetheorem}{S\arabic{theorem}}
\begin{proof}
	We apply the Hadamard test on $W_{lcu}$, giving the quantum circuit $W_p(\bm{x})$ as follows.
	\[
	\Qcircuit @C=1em @R=0.5em {
		\lstick{\ket{0}}& \qw & \gate{H} & \qw  & \ctrl{1} & \qw & \gate{H} & \qw \\
		\lstick{\ket{0}} & {/} \qw  & \qw & \qw &  \multigate{1}{W_{lcu}} & \qw & \qw & \qw \\
		\lstick{\ket{0}} & {/} \qw & \gate{H^{\otimes d}} & \qw  & \ghost{H^{\otimes d}} & \qw & \qw & \qw
	}
	\]
	Measuring the first qubit of $W_{p}(\bm{x})$, we have
	\begin{equation}
		f_{W_p}(\bm{x}) \coloneqq \bra{0} W^\dagger_{p}(\bm{x}) Z^{(0)} W_{p}(\bm{x}) \ket{0} = \bra{0}\bra{+}^{\otimes d} W_{lcu} \ket{0}  \ket{+}^{\otimes d} =  p(\bm{x}). 
	\end{equation}
	The controlled unitary used in LCU,
	\begin{equation}
		U_c(\bm{x}) = \sum_{j=1}^T \ketbra{j}{j} \otimes U^{\bm\a^{(j)}}(\bm{x}),
	\end{equation}
	could be implemented by at most $O(Ts)$ $(\log T)$-qubit controlled gates. A $(\log T)$-qubit controlled gate could be implemented by a quantum circuit consisting of $\CNOT$ gates and single-qubit gates with depth $O(\log T)$~\cite{dasilva2022lineardepth}. Thus $U_c(\bm{x})$ could be implemented by a quantum circuit with depth $O(sT\log T)$ and width $O(d + \log T)$. Then the depth and width of $W_{lcu} = (F^\dagger\otimes I) U_c (F \otimes I)$ are in the same order of $U_c(\bm{x})$ since $F$ is simply tensor of Hadamard gates. Therefore the entire depth of the circuit $W_p$ is $O\bigl(sT\log T + d\bigr)$, and the width of $W_p$ is $O(d + \log T)$. As $T \leq (s+1)d^s$. Note that the number of parameters in the PQC equals the number of parameters in $U_c(\bm{x})$, which is $O(T(s+d))$.
\end{proof}
Note that measuring the first qubit of $W_p(\bm{x})$ for $O(\frac{1}{\eps^2})$ times is needed to estimate the value of $p(\bm{x})$ up to an additive error $\eps$. We could further use the amplitude estimation algorithm~\cite{brassard2002quantum} to reduce the overhead while increasing the circuit depth by $O(\frac{1}{\eps})$.

\subsection{Implementing multivariate trigonometric polynomials}
We extend the PQCs with $R_Z$ encoding to implement multivariate trigonometric polynomials. A multivariate trigonometric polynomials with $d$ variables and degree $s$ is defined as
\begin{equation}
	t(\bm{x}) \coloneqq \sum_{\norm{\bm{n}}_1 \leq s} c_{\bm{n}} e^{i\bm{n} \cdot \bm{x}}
\end{equation}
where $c_{\bm{n}} \in \CC$, $\bm{x} = (x_1, \ldots, x_d) \in \RR^d$, $\bm\n = (\a_1, \ldots, \a_d) \in \ZZ^d$, and $e^{i\bm{n} \cdot \bm{x}} = e^{in_1x_1} e^{in_2x_2}\cdots e^{in_dx_d}$. Consider a trigonometric monomial $c_{\bm{n}} e^{i\bm{n} \cdot \bm{x}} = c_{\bm{n}} e^{in_1x_1} e^{in_2x_2}\cdots e^{in_dx_d}$ such that $\abs{c_{\bm{n}} e^{i\bm{n} \cdot \bm{x}}} \leq 1$ for all $\bm{x} \in \RR^d$ and $\norm{\bm{n}}_1 \leq s$, we could apply the single-qubit PQC with $R_Z$ encoding as defined in \cref{eqn:PQC_z} on each $x_j$ for $1 \leq j \leq d$ respectively.

\begin{lemma}\label{lem:trig_monomial_PQC}
	Given a trigonometric monomial $c_{\bm{n}} e^{i\bm{n} \cdot \bm{x}} = c_{\bm{n}} e^{in_1x_1} e^{in_2x_2}\cdots e^{in_dx_d}$ such that $\abs{c_{\bm{n}} e^{i\bm{n} \cdot \bm{x}}} \leq 1$ for all $\bm{x} \in \RR^d$ and $\norm{\bm{n}}_1 \leq s$, there exists a PQC $U^{\bm{n}}(\bm{x})$ such that
	\begin{equation}
		\bra{0}^{\otimes d}\! U^{\bm{n}}(\bm{x}) \!\ket{0}^{\otimes d} = c_{\bm{n}} e^{i\bm{n} \cdot \bm{x}}.
	\end{equation}
	The width of the PQC is at most $d$, the depth is at most $6s+3$, and the number of parameters is at most $4s+3d$.
\end{lemma}

\begin{proof}
	By \cref{corS:trig_qsp_proj}, we could construct $d$ single-qubit PQCs $U_{\bm\theta_1,\bm\phi_1}^{n_1}(x_1), U_{\bm\theta_2,\bm\phi_2}^{n_2}(x_2), \ldots, U_{\bm\theta_d,\bm\phi_d}^{n_d}(x_d)$ such that
	\begin{align*}
		\braket{0|U_{\bm\theta_1,\bm\phi_1}^{n_1}(x_1)|0} &= c_{\bm{n}}e^{in_1x_1},\\
		\braket{0|U_{\bm\theta_2,\bm\phi_2}^{n_2}(x_2)|0} &= e^{in_2x_2},\\
		\cdots\\
		\braket{0|U_{\bm\theta_d,\bm\phi_d}^{n_d}(x_d)|0} &= e^{in_dx_d},
	\end{align*}
	where the number of layers of each PQC is $L_j = n_j$ for $1 \leq j \leq d$. We then define a $d$-qubit PQC as
	\begin{equation}
		U^{\bm{n}}(\bm{x}) = \bigotimes_{j=1}^d U_{\bm\theta_j,\bm\phi_j}^{n_j}(x_j),
	\end{equation}
	which gives
	\begin{equation}
		\bra{0}^{\otimes d}\! U^{\bm{n}}(\bm{x}) \!\ket{0}^{\otimes d} = \prod_{j=1}^d \braket{0|U_{\bm\theta_j,\bm\phi_j}^{n_j}(x_j)|0} = c_{\bm{n}} e^{i\bm{n} \cdot \bm{x}}.
	\end{equation}
	Since $\norm{\bm{n}}_1 = \sum_{j=1}^d n_j \leq s$, we can conclude that the depth of $U^{\bm{n}}(\bm{x})$ is at most $6s+3$ and the number of parameters in $U^{\bm{n}}(\bm{x})$ is at most $4s + 3d$.
\end{proof}

Then we could apply the technique of LCU on the PQCs $U^{\bm{n}}(\bm{x})$ to implement the operator
\begin{equation}
	U_t(\bm{x}) \coloneqq \sum_{\norm{\bm{n}}_1 \leq s} U^{\bm{n}}(\bm{x}),
\end{equation}
so that we can implement the multivariate trigonometric polynomial as
\begin{equation}
	\bra{+}^{\otimes d}\!  U_t(\bm{x}) \!\ket{+}^{\otimes d} = \sum_{\norm{\bm{n}}_1 \leq s} \bra{+}^{\otimes d}\!U^{\bm{n}}(\bm{x})\!\ket{+}^{\otimes d} = \sum_{\norm{\bm{n}}_1 \leq s} c_{\bm{n}} e^{i\bm{n} \cdot \bm{x}} = t(\bm{x}).
\end{equation}
Note that the number of terms in the summation is
\begin{equation}
	\sum_{\norm{\bm{n}}_1 \leq s} 1 = \sum_{j=0}^{s} \sum_{\norm{\bm{n}}_1=j} 1 \leq \sum_{j=0}^s d^{2s} \leq (s+1)d^{2s}.
\end{equation}
Then, we have the following proposition.
\begin{proposition}\label{prop:PQC_multi_trig}
    For any multivariate trigonometric polynomial $t(\bm{x})$ with $d$ variables and degree $s$ such that $\abs{t(\bm{x})} \leq 1$ for $\bm{x} \in \RR^d$, there exists a PQC $W_{tri}(\bm{x})$ such that
    \begin{equation}
        f_{W_{tri}}(\bm{x}) \coloneqq \bra{0} W^\dagger_{tri}(\bm{x}) Z^{(0)} W_{tri}(\bm{x}) \ket{0} =  t(\bm{x}) 
    \end{equation}
    where $Z^{(0)}$ is the Pauli $Z$ observable on the first qubit. The width of the PQC is $O(d+\log s + s\log d)$, the depth is $O(s^2 d^{2s} (\log s + s\log d))$, and the number of parameters is $O(sd^{2s}(s+d))$.
\end{proposition}
The proof is similar to \cref{prop:pqc_poly}. This result demonstrates the universal approximation property of PQC in the perspective of multivariate Fourier series, which yields similar results as in \citet{schuld2021effect}. Notably, the PQC in \cref{prop:PQC_multi_trig} has an explicit construction without any assumption, improving the implicit PQCs proposed in \citet{schuld2021effect} in terms of circuit size. For instance, to implement the $d$-variable Fourier series with degree $s$, the PQC with parallel structure in \citet{schuld2021effect} requires width $O(ds)$ and potentially exponential depth $O(4^{ds})$.

\section{Approximating continuous functions via PQCs}\label{Appendix:C}
We have constructively shown in the previous section that PQCs could implement multivariate polynomials. To study the approximation capabilities of PQC, a natural strategy involves aggregating multiple polynomials to approximate the continuous function, drawing on well-established principles from classical approximation theory. In the context of univariate functions, this endeavor is guided by the Stone-Weierstrass Theorem~\cite{stone1948generalized}. For the multivariate case, we accomplish this task by employing PQCs to implement Bernstein polynomials, followed by the established result on the approximation error bound of Bernstein polynomials~\cite{heitzinger2002simulation, foupouagnigni2020multivariate}.

\subsection{Established results of Bernstein polynomials approximation}
For a $d$-variable continuous function $f: \RR^d \to \RR$, the multivariate Bernstein polynomial with degree $n\in\NN$ of $f$ is defined as
\begin{equation}\label{eqnS:Bernstein_poly}
	B_n(f; \bm{x}) \coloneqq \sum_{k_1=0}^n \cdots \sum_{k_d=0}^n f\bigl(\frac{\bm{k}}{n}\bigr) \prod_{j=1}^d \binom{n}{k_j} x_j^{k_j}(1-x_j)^{n-k_j},
\end{equation}
and $\bm{k} = (k_1, \ldots, k_d) \in \{0, \ldots, n\}^d$. Then, we have the following lemma on the approximation error bound of the Bernstein polynomial.

\begin{lemma}[Bernstein polynomials approximation for Lipschitz functions~\cite{foupouagnigni2020multivariate}]\label{lemS:Bernstein_poly}
	Given a Lipschitz continuous function $f: [0, 1]^d \to \RR$ with Lipschitz constant $\ell$, which is defined as $|f(\bm{x})-f(\bm{y})|\leq \ell\|\bm{x}-\bm{y}\|_{\infty}$. Let $f$ be bounded by $\Gamma$. The approximation error of the $n$-degree Bernstein polynomial of $f$ scales as 
	\begin{equation}
		\abs{f(\bm{x})-B_n(f;\bm{x})} \leq \eps + 2\Gamma\sum_{j=1}^d \binom{d}{j} \left(\frac{\ell^2}{4n\eps^2}\right)^j \leq \eps + 2\Gamma\left( \left(1+\frac{\ell^2}{4n\eps^2}\right)^d-1 \right)  ,
	\end{equation}
	where $\eps>0$ is an arbitrarily small quantity. 
\end{lemma}
\begin{proof}
Drawing inspiration from the Lipschitz continuity of the target function $f$, we define $\delta = \epsilon/\ell$. Consequently, for any two points $\bm{x}=(x_1, \dots, x_d)$ and $\bm{y}=(y_1, \dots, y_d)$ such that $|x_i-y_i|<\delta$ for all $i\in\{1,\dots,d\}$, it follows that $|f(\bm{x})-f(\bm{y})|\leq \eps$. The target function can be written as 
\begin{eqnarray}
	f(\bm{x})&=&f(x_1, \dots, x_d) \nonumber \\ 
	&=&f\left(x_1, \cdots, x_d\right) \sum_{k_1=0}^{n} \cdots \sum_{k_d=0}^{n} \prod_{i=1}^d\left(\begin{array}{l}
		n \\
		k_i
	\end{array}\right) x_i^{k_i}\left(1-x_i\right)^{n-k_i}  \nonumber \\
	& = &\sum_{k_1=0}^{n} \cdots \sum_{k_d=0}^{n} f\left(x_1, \cdots, x_d\right) \prod_{i=1}^d\left(\begin{array}{l}
		n \\
		k_i
	\end{array}\right) x_i^{k_i}\left(1-x_i\right)^{n-k_i}. 
\end{eqnarray}
Let us consider the set $E=\prod_{i=1}^d\{0,1,\dots , n\}$, and for $j=1,2,\dots,d$, we define the sets
\begin{eqnarray}
\Omega_j=\{k_j\in\{0, 1,\dots,n\}:|\frac{k_i}{n}-x_j|<\delta\} \text{ and } F=E\setminus  (\Omega_1\times\cdots\times\Omega_d). 
\end{eqnarray}
Then, $F=\bigcup_{k=1}^d F_k$, with $F_k=\left\{\prod_{i=1}^d \Omega_{i k}^{\left[\alpha_{i k}\right]} \in F: \alpha_{i k} \in\{0,1\}, \quad \sum_{i=1}^d \alpha_{i k}=k\right\}$, where $\Omega_{i k}^{\left[\alpha_{i k}\right]}=\left\{\begin{array}{ll}\Omega_i & \text { if } \alpha_{i k}=0 \\ \Omega_i^c & \text { if } \alpha_{i k}=1\end{array}\right.$ and $\Omega_i^c=\left\{k_i \in\left\{0, \cdots, n\right\}:\left|\frac{k_i}{n}-x_i\right| \geq \delta\right\}$.
For $A_k=\prod_{i=1}^d \Omega_{i k}^{\left[\alpha_{i k}\right]} \in F_k, k=1, \cdots, d$, let us define also $I_{A_k}=\left\{i \in\{1, \cdots, d\}: \alpha_{i k}=1\right\}$ (that means card $\left.\left(I_{A_k}\right)=k \geq 1\right)$. We have
\begin{equation}
	\begin{aligned}
		& \left|f\left(x_1, \cdots, x_d\right)-B_{n}\left(f; x_1, \cdots, x_d\right)\right| \\
		= & \mid \sum_{k_1=0}^{n} \cdots \sum_{k_d=0}^{n} f\left(x_1, \cdots, x_d\right) \prod_{i=1}^d\left(\begin{array}{l}
			n \\
			k_i
		\end{array}\right) x_i^{k_i}\left(1-x_i\right)^{n-k_i} \\
		& -\sum_{k_1=0}^{n} \cdots \sum_{k_d=0}^{n} f\left(\frac{k_1}{n}, \cdots, \frac{k_d}{n}\right) \prod_{i=1}^d\left(\begin{array}{l}
			n \\
			k_i
		\end{array}\right) x_i^{k_i}\left(1-x_i\right)^{n-k_i} \mid \\
		= & \left|\sum_{k_1=0}^{n} \cdots \sum_{k_d=0}^{n}\left[f\left(x_1, \cdots, x_d\right)-f\left(\frac{k_1}{n}, \cdots, \frac{k_d}{n}\right)\right] \prod_{i=1}^d\left(\begin{array}{l}
			n \\
			k_i
		\end{array}\right) x_i^{k_i}\left(1-x_i\right)^{n-k_i}\right| \\
		\leq & \sum_{k_1=0}^{n} \cdots \sum_{k_d=0}^{n}\left|f\left(x_1, \cdots, x_d\right)-f\left(\frac{k_1}{n}, \cdots, \frac{k_d}{n}\right)\right| \prod_{i=1}^d\left(\begin{array}{l}
			n \\
			k_i
		\end{array}\right) x_i^{k_i}\left(1-x_i\right)^{n-k_i} \\
		\leq & \sum_{\Omega_1} \cdots \sum_{\Omega_d}\left|f\left(x_1, \cdots, x_d\right)-f\left(\frac{k_1}{n}, \cdots, \frac{k_d}{n}\right)\right| \prod_{i=1}^d\left(\begin{array}{l}
			n \\
			k_i
		\end{array}\right) x_i^{k_i}\left(1-x_i\right)^{n-k_i} \\
		& +\sum_F\left|f\left(x_1, \cdots, x_d\right)-f\left(\frac{k_1}{n}, \cdots, \frac{k_d}{n}\right)\right| \prod_{i=1}^d\left(\begin{array}{l}
			n \\
			k_i
		\end{array}\right) x_i^{k_i}\left(1-x_i\right)^{n-k_i }.
	\end{aligned}
\end{equation}
Using the fact that $f$ is continuous and bounded, we get
\begin{equation}\label{eqn:Bernstein_th2}
	\begin{aligned}
		& \left|f\left(x_1, \cdots, x_d\right)-B_{n}\left(f; x_1, \cdots, x_d\right)\right| \\
		\leq & \eps\sum_{\Omega_1} \cdots \sum_{\Omega_d}\prod_{i=1}^d\left(\begin{array}{l}
			n \\
			k_i
		\end{array}\right) x_i^{k_i}\left(1-x_i\right)^{n-k_i} +2\Gamma\sum_F \prod_{i=1}^d\left(\begin{array}{l}
			n \\
			k_i
		\end{array}\right) x_i^{k_i}\left(1-x_i\right)^{n-k_i } \\
	\leq & \eps + 2\Gamma\sum_{l=1}^{d}\sum_{A_l\in F_l} \prod_{i=1}^d\left(\begin{array}{l}
		n \\
		k_i
	\end{array}\right) x_i^{k_i}\left(1-x_i\right)^{n-k_i } \\
	\leq & \eps + 2\Gamma\sum_{l=1}^{d}\sum_{A_l\in F_l} \prod_{i\in I_{A_l}}\frac{1}{4n\delta^2} \\
	= & \eps + 2\Gamma\sum_{j=1}^d \binom{d}{j} \frac{1}{(4n\delta^2)^j} \leq \eps + 2\Gamma\left( \left(1+\frac{\ell^2}{4n\eps^2}\right)^d-1 \right).
	\end{aligned}
\end{equation}
This completes the proof. A more detailed expansion of \cref{eqn:Bernstein_th2} can be seen in Theorem~2 in \citet{foupouagnigni2020multivariate}.
\end{proof}
\begin{remark}\label{Remark:Continuous target Bernstein}
Here, it is important to observe that for a continuous target function, denoted as $f(\bm{x})$, there exists a value of $\delta>0$ such that: $$\left|f\left(x_1, \cdots, x_d\right)-B_{n}\left(f; x_1, \cdots, x_d\right)\right|\leq \eps + 2\Gamma\left( \left(1+\frac{1}{4n\delta^2}\right)^d-1 \right).$$ This expression signifies the convergence rate of the Bernstein polynomial for general continuous functions.
\end{remark}

\subsection{Implement Bernstein polynomials via PQCs}
In Lemma~\ref{lemS:Bernstein_poly}, we have defined the Bernstein polynomial and its approximation error towards the Lipschitz continuous function. Guided by~\cref{prop:pqc_poly}, we can construct a PQC to implement such a Bernstein polynomial.
\begin{lemma}\label{lem:app_Bernstein_poly}
	For any $d$-variable Bernstein polynomial  with degree $n\in\NN$ defined in \cref{eqnS:Bernstein_poly} such that $\abs{B_n(f;\bm{x})} \leq 1$ for $\bm{x}\in[0, 1]^d$, there exist a PQC $W_b(\bm{x})$ satisfying
    \begin{equation}
        f_{W_b}(\bm{x}) \coloneqq \bra{0} W^\dagger_{b}(\bm{x}) Z^{(0)} W_{b}(\bm{x}) \ket{0} = B_n(f;\bm{x}).
    \end{equation}
    The width of the PQC is $O(d\log{n})$, the depth is $O\bigl(dn^{d}\log{n}\bigr)$, and the number of parameters is $O(dn^d)$.
\end{lemma}
\begin{proof}
	We undertake a two-step process in the proof of \cref{lem:app_Bernstein_poly}. Initially, we construct PQCs to provide an exact representation of $f\bigl(\frac{\bm{k}}{n}\bigr) \prod_{j=1}^d \binom{n}{k_j} x_j^{k_j}(1-x_j)^{n-k_j}$ for all $\bm{k}\in\{0, 1, \dots, n\}^d$. Subsequently, we employ LCU to aggregate these PQCs for the purpose of approximating the Bernstein polynomial described in \cref{eqnS:Bernstein_poly}.
	
	The univariate polynomial $x^k(1-x)^{n-k}$ can be represented by a PQC. The depth of this PQC is less than $2n+1$, the width is $2$, and the number of parameters is $n+2$. The multivariate polynomial $f\bigl(\frac{\bm{k}}{n}\bigr) \prod_{j=1}^d \binom{n}{k_j} x_j^{k_j}(1-x_j)^{n-k_j}$ can be exactly represented by the product of the univariate polynomial $x^k(1-x)^{n-k}$. The same routine has been employed in Lemma~\ref{lem:monomial_PQC}. The depth of this PQC is less than $2n+1$, the width is $2d$, and the number of parameters is $d(n+2)$.
	
	The number of terms in the summation in \cref{eqnS:Bernstein_poly} is $(n+1)^d$. We can employ the same routine in~\cref{prop:pqc_poly} to construct the PQC $W_b(\bm{x})$. The depth of $W_b$ scales as
	\[O\Bigl(\bigl(d(n+1)^{d+1}\log{(n+1)} + d\bigr)\Bigr),\]
	the width is $2d+d\log{(n+1)}$, and the number of parameters is $(n+1)^d(n+2)d$. The results presented in Lemma~\ref{lem:app_Bernstein_poly} can be obtained after simplification.
\end{proof}

\subsection{PQC approximating continuous functions}

We have successfully derived results regarding the approximation error between PQCs and Bernstein polynomials and between Bernstein polynomials and continuous functions. Leveraging these established findings, we can now formulate a rigorous assertion regarding the universal approximation theorem and the error bound of PQCs, employing the well-established principles of triangle inequality.

\renewcommand\thetheorem{\ref{thm:pqc_uat}}
\setcounter{theorem}{\arabic{theorem}-2}
\begin{theorem}[The Universal Approximation Theorem of PQC]
    For any continuous function $f:[0,1]^d \to [-1,1]$, given an $\eps > 0$, there exist an $n\in \NN$ and a PQC $W_b(\bm{x})$ with width $O(d\log n)$, depth $O(dn^d\log n)$ and the number of trainable parameters $O(dn^d)$ such that
    \begin{equation}
        \abs{f(\bm{x}) - f_{W_b}(\bm{x})} \leq \eps
    \end{equation}
    for all $\bm{x}\in [0,1]^d$, where $f_{W_b}(\bm{x}) \coloneqq \bra{0} W^\dagger_{b}(\bm{x}) Z^{(0)} W_{b}(\bm{x}) \ket{0}$.
\end{theorem}
\renewcommand{\thetheorem}{S\arabic{theorem}}

\begin{proof}
\Cref{Remark:Continuous target Bernstein} has established the uniform convergence of the Bernstein polynomial towards any continuous function within the cubic domain $[0, 1]^d$, denoted as $B_n(f;\bm{x})$, with the property that $B_n(f;\bm{x})\rightarrow f(\bm{x})$ as $n\rightarrow +\infty$. Building on Lemma~\ref{lem:app_Bernstein_poly}, we can effectively implement this Bernstein polynomial $B_n(f;\bm{x})$ using $f_{W_b}(\bm{x})$. The depth of the PQC $W_b(\bm{x})$ is $O\bigl(dn^{d}\log{n}\bigr)$, the width is $O(d\log{n})$, and the number of parameters is $O(dn^d)$. This completes the proof.
\end{proof}

\renewcommand\thetheorem{\ref{thm:approx_Lipschitz}}
\setcounter{theorem}{\arabic{theorem}-3}
\begin{theorem}
Given a Lipschitz continuous function $f: [0, 1]^d \to [-1, 1]$ with a Lipschitz constant $\ell$, for any $\eps > 0$ and $n \in \NN$, there exists a PQC $W_b(\bm{x})$ with such that $f_{W_b}(\bm{x}) \coloneqq \bra{0} W^\dagger_{b}(\bm{x}) Z^{(0)} W_{b}(\bm{x}) \ket{0}$ satisfies
\begin{equation}
    \abs{f(\bm{x})-f_{W_b}(\bm{x})} \leq \eps + 2\biggl( \Bigl(1+\frac{\ell^2}{n\eps^2}\Bigr)^d-1 \biggr) \leq \eps + d2^{d}\frac{\ell^2}{n\eps^2}
\end{equation}
for all $\bm{x}\in[0, 1]^d$. The width of the PQC is $O(d\log n)$, the depth is $O\bigl(dn^{d}\log{n}\bigr)$, and the number of parameters is $O(dn^d)$.
\end{theorem}
\renewcommand{\thetheorem}{S\arabic{theorem}}

\begin{proof}
\Cref{lemS:Bernstein_poly} has established the uniform convergence rate of the Bernstein polynomial towards any Lipschitz continuous function within the cubic domain $[0, 1]^d$. We know that for any Lipschitz continuous function $f(\bm{x})$ with Lipschitz constant $\ell$, there exists a Bernstein polynomial $B_n(f;\bm{x})$ satisfying $$\abs{f(\bm{x})-B_n(f;\bm{x})} \leq \eps + 2\Gamma\sum_{j=1}^d \binom{d}{j} \left(\frac{\ell^2}{4n\eps^2}\right)^j \leq \eps + 2\Gamma\left( \left(1+\frac{\ell^2}{4n\eps^2}\right)^d-1 \right).$$ Building on Lemma~\ref{lem:app_Bernstein_poly}, we can effectively implement this Bernstein polynomial $B_n(f;\bm{x})$ using $f_{W_b}(\bm{x})$. The depth of the PQC $W_b(\bm{x})$ is $O\bigl(dn^{d}\log{n}\bigr)$, the width is $O(d\log{n})$, and the number of parameters is $O(dn^d)$. This completes the proof.
\end{proof}

\section{Approximating smooth functions via nested PQCs}\label{Appendix:D}
Other than using a Bernstein polynomial to approximate a continuous function globally, we could also utilize local polynomials to achieve a piecewise approximation. To do this, we follow the path of classical deep neural networks~\cite{petersen2018optimal, lu2021deep, jiao2023deep}, using multivariate Taylor series expansion to approximate a multivariate smooth function $f$ in some small local region. Let $\beta = s + r >0$, $r = (0, 1]$ and $s=\floor{\beta}\in \NN$, for a finite constant $B_0 > 0$, the $\beta$-$\holder$ class of functions $\cH^\beta([0,1]^d, B_0)$ is defined as
\begin{equation}
    \cH^\beta([0,1]^d, B_0) = \Bigl\{ f:[0,1]^d \to \RR, \max_{\norm{\bm\a}_1\leq s} \norm{\partial^{\bm\a}f}_{\infty} \leq B_0, \max_{\norm{\bm\a}_1=s} \sup_{\bm{x} \neq \bm{y}} \frac{\abs{\partial^{\bm\a}f(\bm{x})-\partial^{\bm\a}f(\bm{y})}}{\norm{\bm{x}-\bm{y}}_2^r} \leq B_0\Bigr\},
\end{equation}
where $\partial^{\bm\a} = \partial^{\a_1} \cdots \partial^{\a_d}$ for $\bm\a=(\a_1,\ldots, \a_d)\in \NN^d$. By definition, for a function $f\in \cH^\beta([0,1]^d, B_0)$, when $\beta\in(0,1)$, $f$ is a $\holder$ continuous function with order $\beta$ and $\holder$ constant $B_0$; when $\beta = 1$, $f$ is a Lipschitz function with Lipschitz constant $B_0$; when $\beta > 1$, $f$ belongs to the $C^s$ class of functions whose $s$-th partial derivatives exist and are bounded.

We utilize the following lemma on the Taylor expansion of $\beta$-$\holder$ functions as a mathematical tool for constructing and analyzing the PQC approximation.
\begin{lemma}[\cite{petersen2018optimal}]\label{lemS:Taylor_expansion_holder}
    Given a function $f \in \cH^\beta([0,1]^d, 1)$ with $\beta = r + s$, $r\in(0,1]$ and $s\in \NN^{+}$, for any $\bm{x}, \bm{x_0} \in [0,1]^d$, we have
    \begin{equation}\label{eqnS:taylor_expansion}
    \abs[\Big]{f(\bm{x}) - \sum_{\norm{\bm{\a}}_1 \leq s} \frac{\partial^{\bm{\a}} f(\bm{x_0})}{\bm\a !} (\bm{x} - \bm{x_0})^{\bm{\a}}} \leq d^s \norm{\bm{x} - \bm{x_0}}^\beta_2,
    \end{equation}
    where $\bm\a! = \a_1!\cdots \a_d!$.
\end{lemma}
Next, we show how to construct PQCs to implement the Taylor expansion of $\beta$-$\holder$ functions.

\subsection{Localization via PQC}
As shown in \cref{eqnS:taylor_expansion}, the Taylor expansion of a multivariate smooth function only converges in a fairly small local region. So, we need first to localize the entire region $[0, 1]^d$. Given $K \in \NN$ and $\Delta \in (0, \frac{1}{3K})$, for each $\bm{\eta} = (\eta_1, \ldots, \eta_d) \in \{0, 1, \ldots, K-1\}^d$, we define
\begin{equation}
	Q_{\bm{\eta}} \coloneqq \Bigl\{ \bm{x}=(x_1,\ldots,x_d): x_i \in \bigl[ \frac{\eta_i}{K}, \frac{\eta_i+1}{K} - \Delta \cdot 1_{\eta_i < K - 1}\bigr] \Bigr\}.
\end{equation}
By the definition of $Q_{\bm{\eta}}$, the region $[0, 1]^d$ is approximately divided into small hypercubes $\bigcup_{\bm\eta}Q_{\bm{\eta}}$ and some trifling region $\Lambda(d, K, \Delta) \coloneqq [0,1]^d \setminus (\bigcup_{\bm\eta}Q_{\bm{\eta}})$, as illustrated in \cref{fig:fig2} in the main text. 
Then we need to construct a PQC that maps any $x \in Q_{\bm{\eta}}$ to some fixed point $x_{\bm{\eta}} = \frac{\bm\eta}{K} \in Q_{\bm{\eta}}$, i.e., approximating the piecewise-constant function $D(\bm{x}) = \frac{\bm\eta}{K}$ if $\bm{x} \in Q_{\bm\eta}$, where $\frac{\bm{\eta}}{K} = (\eta_1/K, \ldots, \eta_d/K)$. We consider the case of $d=1$, where the localization function is
\begin{equation}\label{eqn:piecewise_function}
	D(x) = \frac{k}{K}, \qquad \text{if $x \in \Bigl[\frac{k}{K}, \frac{k+1}{K} - \Delta \cdot 1_{k < K - 1}\Bigr]$ for $k=0, 1, \ldots, K-1$}.
\end{equation}
The multivariate case could be easily generalized by applying $D(x)$ to each variable $x_j$.
The idea is to find a polynomial that approximates the sign function
\begin{equation}
	\sgn(x-c) = \begin{cases}
		1, & \text{if $x > c$,}\\[1pt]
		0, & \text{if $x = c$}\\
		-1, & \text{if $x < c$}
	\end{cases},
\end{equation}
as shown in the following lemma.
\begin{lemma}[Polynomial approximation to the sign function $\sgn(x-c)$~\cite{low2017quantum}]\label{lem:sign_function_approx}
	$\forall c\in[-1, 1], \Delta > 0, \eps \in (0, 1)$. there exists an odd polynomial $P_{\Delta, \eps}(x)$ of degree $n=O(\frac{1}{\Delta}\log \frac{1}{\eps})$ that satisfies
	\begin{enumerate}
		\item $\abs{P_{\Delta, \eps}(x-c)} \leq 1$ for all $x\in [-1, 1],$
		\item $\abs{\sgn(x-c) - P_{\Delta, \eps}(x-c)} \leq \eps $ for all $x \in [-1, 1] \setminus (c - \frac{\Delta}{2}, c + \frac{\Delta}{2})$.
	\end{enumerate}
\end{lemma}
Note that we could also approximate the step function defined as $\stp(x - c) \coloneqq \frac{1}{2} \sgn(x-c) + \frac{1}{2}$ by the polynomial $P_{\Delta, \eps}'(x-c) = \frac{1}{2} P_{\Delta, \eps}(x-c) + \frac{1}{2}$ of degree $n=O(\frac{1}{\Delta}\log \frac{1}{\eps})$, which satisfies that $\abs{P_{\Delta, \eps}'(x-c)} \leq 1$ for all $x\in [-1, 1]$ and $\abs{\stp(x-c) - P_{\Delta, \eps}'(x-c)} \leq \frac{\eps}{2} $ for all $x \in [-1, 1] \setminus (c - \frac{\Delta}{2}, c + \frac{\Delta}{2})$. Note that the polynomial $P_{\Delta, \eps}'(x-c)$ does not have definite parity and thus cannot be directly implemented by a PQC as shown in \cref{cor:qsp}. Since only the domain $[0,1]$ is relevant to $x$, for $c\in(0,1)$, we could define an even polynomial
\begin{equation}
	P_{c, \Delta, \eps}^{\text{even}}(x) = \frac{1}{1 + \frac{\eps}{2}}\mleft(P_{\Delta, \eps}'(x-c) + P_{\Delta, \eps}'(-x-c) \mright)
\end{equation}
such that $\abs{P_{c, \Delta, \eps}^{\text{even}}(x)} \leq 1$ for all $x\in [-1, 1]$ and $\abs{\stp(x-c) - P_{c,\Delta, \eps}^{\text{even}}(x)} \leq \frac{\eps}{2} $ for all $x \in [0, 1] \setminus (c - \frac{\Delta}{2}, c + \frac{\Delta}{2})$. The piecewise-constant function $D(x)$ can be written as a combination of step functions,
\begin{equation}\label{eqn:piecewise_as_step}
	D(x) = \sum_{k=1}^{K-1} \frac{1}{K} \stp\bigl(x-\frac{k}{K} + \frac{\Delta}{2}\bigr).
\end{equation}
Then we could find even polynomials $P_{c, \Delta, \eps}^{\text{even}}(x)$ that approximate $\stp\bigl(x-\frac{k}{K} + \frac{\Delta}{2}\bigr)$ for each $k$. Combining those polynomials together as in \cref{eqn:piecewise_as_step}, we have the following lemma.
\begin{lemma}\label{lemS:piecewise_poly}
	Given $K \in \NN$ and $\Delta \in (0, \frac{1}{3K})$, there exists an even polynomial $P_{\Delta, \eps}(x)$ of degree $n= O(\frac{1}{\Delta}\log\frac{K}{\eps})$ that satisfies
	\begin{enumerate}
		\item $\abs{P_{\Delta, \eps}(x)} \leq 1$ for all $x\in [-1, 1],$
		\item $\abs{D(x) - P_{\Delta, \eps}(x)} \leq \eps $ for all $x \in \bigcup_{k=0}^{K-1} \bigl[\frac{k}{K}, \frac{k+1}{K} - \Delta \cdot 1_{k < K - 1}\bigr]$.
	\end{enumerate}
\end{lemma}
Note that we could shift the polynomial $P_{\Delta, \eps}(x)$ such that $P_{\Delta, \eps}(x) - D(x) \in (0, \eps)$ without changing the degree. It follows that we can construct a PQC to implement the polynomial $P_{\Delta, \eps}(x)$ by \cref{cor:qsp}.
\begin{corollary}\label{cor:localization}
	Given $K \in \NN$, $\Delta \in (0, \frac{1}{3K})$ and $\eps \in (0, \frac{1}{K})$, there exists a single-qubit PQC $U_{D}(x)$ of depth $O(\frac{1}{\Delta}\log\frac{K}{\eps})$ that satisfies
	\begin{equation}
		\braket{+| U_{D}(x) |+} - \frac{k}{K}\in (0, \eps) \quad \text{if $x \in \Bigl[\frac{k}{K}, \frac{k+1}{K} - \Delta \cdot 1_{k < K - 1}\Bigr]$ for $k=0, 1, \ldots, K-1$}.
	\end{equation}
\end{corollary}
Note that $\eps$ has to be bounded by $\frac{1}{K}$, which is the length of the localized region. We could further implement such a localization procedure for $\bm{x} = (x_1, \ldots, x_d)$ on the region $[0,1]^d$ by applying the PQC for each $x_j$, as stated in the following corollary.
\begin{lemma}[Localization via PQC]\label{lem:localization_PQC}
	Given $K \in \NN$, $\Delta \in (0, \frac{1}{3K})$ and $\eps \in (0, \frac{1}{K})$, there exists a PQC $W_{D}(\bm{x})$ of width $O(d)$ and depth $O(\frac{1}{\Delta}\log\frac{K}{\eps})$ implementing a localization function $f_{W_D}(\bm{x}): \RR^d \to \RR^d$ such that
	\begin{equation}\label{eqn:piecewise_PQC}
		\bm{0} \leq f_{W_D}(\bm{x}) - \frac{{\bm{\eta}}}{K} \leq \bm{\eps} \quad \text{if $\bm{x} \in Q_{\bm{\eta}}$,}
	\end{equation}
	where $\bm{0} =(0, \ldots, 0)$ and $\bm{\eps} =(\eps, \ldots, \eps)$ are $d$-dimensional vectors.
\end{lemma}
\begin{proof}
	We construct a $d$-qubit PQC $W_{D}(\bm{x}) \coloneqq \bigotimes_{j=1}^d U_D(x_j)$ where the single-qubit PQC $U_D(x)$ is constructed in \cref{cor:localization}.
	Then we apply the Hadamard test on each $U_D(x_j)$ to obtain $f_{U_D}(x_j) \coloneqq \braket{+| U_{D}(x_j) |+}$. Let $f_{W_D}(\bm{x}) \coloneqq (f_{U_D}(x_1), \ldots, f_{U_D}(x_d))$, which implements the localization function as required.
\end{proof}

\subsection{Implementing the Taylor coefficients by PQC}
Next, we use PQC to implement the Taylor coefficients, which is essentially a point-fitting problem. For each ${\bm{\eta}} = (\eta_1, \ldots, \eta_d) \in \{0, 1, \ldots, K-1\}^d$ and $\bm\a$, we denote $\xi_{{\bm{\eta}}, \bm\a} \coloneqq \frac{\partial^{\bm\a} f(\frac{{\bm{\eta}}}{K})}{\bm\a !} \in [-1, 1]$. Then we could construct the following PQC,
\begin{equation}\label{eqn:Taylor_coeff_PQC}
	U_{co}^{\bm\a} = \sum_{{\bm{\eta}}} \ketbra{{\bm{\eta}}}{{\bm{\eta}}} \otimes R_X(\theta_{\bm{\eta},\bm \a}),
\end{equation}
where $\ket{{\bm{\eta}}} = \ket{\eta_1} \otimes \cdots \otimes \ket{\eta_d}$ and $\theta_{\bm \eta,\bm \a}=2\arccos(\xi_{\bm\eta,\bm\a})$. It gives the following lemma.

\begin{lemma}\label{lem:taylor_coeff_PQC}
	Given a $\beta$-H{\"o}lder smooth function $f:[0,1]^d \to [-1,1]$, for any $\bm{\a} \in \NN^d$ and ${\bm{\eta}} \in \{0, 1, \ldots, K-1\}^d$, there exists a PQC $U_{co}^{\bm\a}$ such that
	\begin{equation}
		\bra{{\bm{\eta}}, 0} U_{co}^{\bm\a} \ket{{\bm{\eta}}, 0} = \xi_{{\bm{\eta}}, \bm \alpha}.
	\end{equation}
	The width of the PQC is $O(d\log K)$, and the depth is $O(K^d)$.
\end{lemma}
We note that the state $\ket{{\bm{\eta}}}$ can be prepared using basis encoding according to the results of localization in \cref{lem:localization_PQC}.

\subsection{Implementing multivariate Taylor series by PQC}
To implement the multivariate Taylor expansion of a function, we first build a PQC to represent a single term in the Taylor series, which could be done by combining the monomial implementation in \cref{lem:monomial_PQC} and the Taylor coefficient implementation in \cref{lem:taylor_coeff_PQC}. Thus, we have the following corollary.
\begin{corollary}\label{cor:taylor_term}
	For any $\beta$-H{\"o}lder smooth function $f$, given an $\bm{\a} \in \NN^d$ with $\norm{\bm{\a}}_1 \leq s$ for $s\in\NN^{+}$ and an ${\bm{\eta}} \in \{0, 1, \ldots, K-1\}^d$, there exists a PQC $U^{\bm\a}_{\bm\eta}(\bm{x})$ such that
	\begin{equation}
		\bra{\bm\eta, 0}\!\bra{+}^{\otimes d} U^{\bm\a}_{\bm\eta}(\bm{x}) \ket{\bm\eta, 0}\!\ket{+}^{\otimes d} = \frac{\partial^{\bm{\a}} f(\frac{\bm{\eta}}{K})}{\bm\a !} \bigl(\bm{x} - \frac{\bm{\eta}}{K}\bigr)^{\bm{\a}}.
	\end{equation}
	The width of the PQC is $O(d\log K)$, the depth is $O(K^d + s)$, and the number of parameters is at most $K^d+s+d$.
\end{corollary}
\begin{proof}
	Let $U^{\bm\a}_{\bm\eta}(\bm{x}) \coloneqq U_{co}^{\bm{\a}} \otimes U^{\bm\a}(\bm{x} - \frac{\bm{\eta}}{K})$, where $U_{co}^{\bm{\a}}$ is defined in \cref{lem:taylor_coeff_PQC} and $U^{\bm\a}(\bm{x} - \frac{\bm{\eta}}{K})$ is defined in \cref{lem:monomial_PQC} with changing input from $\bm{x}$ to $\bm{x} - \frac{\bm\eta}{K}$. Then the corollary directly follows from \cref{lem:monomial_PQC} and \cref{lem:taylor_coeff_PQC}.
\end{proof}

The next step is to combine single Taylor terms together to implement the truncated Taylor expansion of the target function. The method is in the same spirit as what is utilized in~\cref{prop:pqc_poly}, i.e., using LCU to achieve the following (unnormalized) operator,
\begin{equation}
	U_t(\bm{x}) \coloneqq \sum_{\norm{\bm\a}_1 \leq s} U^{\bm\a}_{\bm\eta}(\bm{x}).
\end{equation}
Then we can implement the Taylor expansion of the function $f$ at point $\frac{\bm\eta}{K}$ as
\begin{equation}
	\bra{\bm\eta, 0}\!\bra{+}^{\otimes d} U_t(\bm{x}) \ket{\bm\eta, 0}\!\ket{+}^{\otimes d} = \sum_{\norm{\bm\a}_1 \leq s} \frac{\partial^{\bm{\a}} f(\frac{\bm{\eta}}{K})}{\bm\a !} \bigl(\bm{x} - \frac{\bm{\eta}}{K}\bigr)^{\bm{\a}}.
\end{equation}
Hence we have the following lemma.

\begin{lemma}\label{lem:taylor_series_PQC}
	Given a function $f \in \cH^\beta([0,1]^d, 1)$ with $\beta = r + s$, $r\in(0,1]$ and $s\in \NN^{+}$, for any $\bm\eta \in \{0, \ldots, K-1\}^d$, there exists a PQC $W_e(\bm{x}, \frac{\bm{\eta}}{K})$ such that $f_{W_e}(\bm{x}) \coloneqq \bra{0} W^\dagger_{e}(\bm{x}) Z^{(0)} W_{e}(\bm{x}) \ket{0}$ implements the truncated Taylor expansion at point $\frac{\bm\eta}{K}$,
	\begin{equation}
		f_{W_e}(\bm{x}) = \sum_{\norm{\bm\a}_1 \leq s} \frac{\partial^{\bm{\a}} f(\frac{\bm{\eta}}{K})}{\a !} \bigl(\bm{x} - \frac{\bm{\eta}}{K}\bigr)^{\bm{\a}}.
	\end{equation}
	The depth of the PQC is $O(s^2 d^sK^d(\log s + s\log d + d\log K))$, the width is $O(d\log K+\log s + s\log d)$, and the number of parameters is $O(sd^s(s+d+K^d))$.
\end{lemma}
\begin{proof}
	The idea of constructing the PQC $W_e(\bm{x}, \frac{\bm\eta}{K})$ is similar to the construction of $W_p(\bm{x})$ in \cref{prop:pqc_poly}. The only difference is that here we apply LCU on unitaries $U^{\bm\a}_{\bm\eta}(\bm{x}) \coloneqq U_{co}^{\bm{\a}} \otimes U^{\bm\a}(\bm{x} - \frac{\bm{\eta}}{K})$ instead of $U^{\bm\a}(\bm{x})$. Thus, the controlled unitary is \begin{equation}
		U_c\bigl(\bm{x}, \frac{\bm\eta}{K}\bigr) = \sum_{j=1}^T \ketbra{j}{j} \otimes U^{\bm\a^{(j)}}_{\bm\eta}(\bm{x})
	\end{equation}
	and the unitary $W_{lcu}(\bm{x}, \frac{\bm\eta}{K}) = (F^\dagger\otimes I) U_c(\bm{x}, \frac{\bm\eta}{K}) (F \otimes I)$ satisfies that
	
	\begin{equation}
		\bra{0}\!\bra{\bm\eta, 0}\!\bra{+}^{\otimes d} W_{lcu}\bigl(\bm{x}, \frac{\bm\eta}{K}\bigr)\ket{0}\! \ket{\bm\eta, 0}\!\ket{+}^{\otimes d} = \sum_{\norm{\bm\a}_1 \leq s} \frac{\partial^{\bm{\a}} f(\frac{\bm{\eta}}{K})}{\bm\a !} (\bm{x} - \frac{\bm{\eta}}{K})^{\bm{\a}}.
	\end{equation}
	We then apply the Hadamard test on $W_{lcu}(\bm{x}, \frac{\bm\eta}{K})$, giving the quantum circuit $W_e(\bm{x}, \frac{\bm\eta}{K})$ as below
	\[
	\Qcircuit @C=1em @R=0.5em {
		\lstick{\ket{0}} & \qw & \gate{H} & \qw  & \ctrl{1} & \qw & \gate{H} & \qw \\
		\lstick{\ket{0}} & {/} \qw  & \qw & \qw &  \multigate{3}{W_{lcu}} & \qw & \qw & \qw \\
		\lstick{\ket{0}} & {/} \qw  & \gate{U(\bm\eta)} & \qw &  \ghost{W_{lcu}} & \qw & \qw & \qw \\
		\lstick{\ket{0}} & \qw  & \qw & \qw &  \ghost{W_{lcu}} & \qw & \qw & \qw \\
		\lstick{\ket{0}} & {/} \qw & \gate{H^{\otimes d}} & \qw  & \ghost{W_{lcu}} & \qw & \qw & \qw
	}
	\]
	where the unitary $U(\bm\eta)$ takes $\bm\eta$ as input and maps $\ket{0}$ to $\ket{\bm\eta}$. Note that the controlled unitary $U_c(\bm{x}, \frac{\bm\eta}{K})$ could be implemented by $O(T(s+1))$ number of $(\log T)$-qubit controlled gates and $O(TK^d)$ number of $(\log T + d\log K)$-qubit controlled gates. An $n$-qubit controlled gate could be implemented by a quantum circuit consisting of $\CNOT$ gates and single-qubit gates with depth $O(n)$~\cite{dasilva2022lineardepth}. Thus $U_c(\bm{x})$ could be implemented by a quantum circuit with depth $O((s+1)T\log T + TK^d (\log T + d\log K))$ and width $O(d + \log T + d\log K)$. Then the depth and width of $W_{lcu}(\bm{x}, \frac{\bm\eta}{K}) = (F^\dagger\otimes I) U_c(\bm{x}, \frac{\bm\eta}{K}) (F \otimes I)$ are in the same order of $U_c(\bm{x}, \frac{\bm\eta}{K})$ since $F$ is simply tensor of Hadamard gates. Therefore the entire depth of the circuit $W_e$ is $O((sT\log T + TK^d (\log T + d\log K)))$ and the width is $O(d + \log T + d\log K)$. As $T \leq (s+1)d^s$, we have the depth and width of PQC shown in~\cref{lem:taylor_series_PQC}. Note that the number of parameters in the PQC equals the number of parameters in $U_c(\bm{x})$, which is $O(T(s+d + K^d))$.
\end{proof}

Finally, we combine the steps of localization and the Taylor series implementation to achieve a local Taylor expansion for the target function. The PQC is in a nested structure consisting of a PQC for localization and a PQC for Taylor series; see the detailed construction in the following theorem.
\renewcommand\thetheorem{\ref{thm:approx_holder}}
\setcounter{theorem}{\arabic{theorem}-1}
\begin{theorem}
    Given a function $f \in \cH^\beta([0,1]^d, 1)$ with $\beta = r + s$, $r\in(0,1]$ and $s\in \NN^{+}$, for any $K\in\NN$ and $\Delta\in(0, \frac{1}{3K})$, there exists a PQC $W_{t}(\bm{x})$ such that $f_{W_t}(\bm{x}) \coloneqq \bra{0} W^\dagger_{t}(\bm{x}) Z^{(0)} W_{t}(\bm{x}) \ket{0}$ satisfies
    \begin{equation}
        \abs{f(\bm{x}) - f_{W_t}(\bm{x})} \leq d^{s+\beta/2}K^{-\beta}
    \end{equation}
    for $\bm{x} \in \bigcup_{{\bm{\eta}}} Q_{\bm{\eta}}$. The width of the PQC is $O(d\log K + \log s + s\log d)$, the depth is $O(s^2 d^sK^d(\log s + s\log d + d\log K)) + \frac{1}{\Delta}\log K)$, and the number of parameters is $O(sd^s(s+d+K^d) + \frac{d}{\Delta} \log K)$.
\end{theorem}
\renewcommand{\thetheorem}{S\arabic{theorem}}
\begin{proof}
	By \cref{lemS:Taylor_expansion_holder}, we have the following error bound for $\bm{x} \in Q_{\bm{\eta}}$,
	\begin{equation}\label{eqn:taylor_error}
		\abs[\Big]{f(\bm{x}) - \sum_{ \norm{\bm\a}_1 \leq s} \frac{\partial^{\bm\a} f(\frac{{\bm{\eta}}}{K})}{\bm\a !} (\bm{x}-\frac{{\bm{\eta}}}{K})^{\bm\a}} \leq d^s \norm[\big]{\bm{x}-\frac{{\bm{\eta}}}{K}}^\beta_2 \leq d^{s+\beta/2}K^{-\beta}.
	\end{equation}
	Motivated by this, we first construct a localization PQC $W_D(x)$ as in \cref{lem:localization_PQC} such that
	\begin{equation}
		\bm{0} \leq f_{W_D}(\bm{x}) - \frac{\bm\eta}{\bm{K}} \leq \Bigl(\frac{1}{2K}, \ldots, \frac{1}{2K}\Bigr) \quad \text{if $\bm{x} \in Q_{\bm{\eta}}$}.
	\end{equation}
	The depth of $W_D(x)$ is $O(\frac{1}{\Delta}\log K)$. We then construct a PQC
	\begin{equation}
		W_t(\bm{x}) \coloneqq W_e(\bm{x}, f_{W_D}(\bm{x})),
	\end{equation}
	where $W_e$ is the PQC proposed in \cref{lem:taylor_series_PQC}. Note that the state $\ket{\bm\eta}$ in \cref{lem:taylor_series_PQC} could be prepared by rounding $f_{W_D}(\bm\eta) K$, i.e., $\bm\eta = \floor{f_{W_D}(\bm\eta) K}$. In other words, the PQC $W_t(\bm{x})$ has a nested structure consisting of a PQC for localization and a PQC for Taylor series implementation. Then we show that $f_{W_t}(\bm{x}) \coloneqq \bra{0} W^\dagger_{t}(\bm{x}) Z^{(0)} W_{t}(\bm{x}) \ket{0}$ can approximate $\beta$-H{\"o}lder smooth function $f$ on $\bigcup_{\bm{\eta}} Q_{\bm{\eta}}$. By the triangle inequality and \cref{eqn:taylor_error}, we have
	\begin{align}
		\abs{f(\bm{x}) - f_{W_t}(\bm{x})} &\leq \abs[\Big]{f_{W_t}(\bm{x})  - \sum_{\norm{\bm\a}_1 \leq s} \frac{\partial^{\bm\a} f(f_{W_D}(\bm{x}))}{\bm\a !} (x-f_{W_D}(\bm{x}))^{\bm\a}} + d^s\norm{\bm{x} - f_{W_D}(\bm{x})}_2^\beta\\
		&\leq \abs[\Big]{f_{W_t}(\bm{x})  - \sum_{\norm{\bm\a}_1 \leq s} \frac{\partial^{\bm\a} f(f_{W_D}(\bm{x}))}{\bm\a !} (x-f_{W_D}(\bm{x}))^{\bm\a}} + d^{s+\beta/2}K^{-\beta}\\
		&\leq d^{s+\beta/2}K^{-\beta}.
	\end{align}
	The second inequality comes from the fact that $||\bm{x} - f_{W_D}(\bm{x})||_2 \leq \frac{1}{K}$ for $\bm{x} \in Q_{\bm{\eta}}$. This completes the proof.
\end{proof}

Note that the PQC in \cref{thm:approx_holder} is nesting of two PQCs, while its depth is counted as the sum of two PQCs for simplicity. We have established the uniform convergence property of PQCs for approximating $\holder$ smooth function on $[0, 1]^d$ except for the trifling region $\Lambda(d, K, \Delta)$. Note that the Lebesgue measure of such a trifling region is no more than $dK\Delta$. We can set $\Delta=K^{-d}$ with no influence on the size of the constructed PQC in \cref{thm:approx_holder}. Since $\nu$ is absolutely continuous with respect to the Lebesgue measure, we have the following corollary.
\begin{corollary}
	Given a function $f \in \cH^\beta([0,1]^d, 1)$ with $\beta = r + s$, $r\in(0,1]$ and $s\in \NN^{+}$, for any $K\in\NN$ and $\Delta\in(0, \frac{1}{3K})$, there exists a PQC $W_{t}(\bm{x})$ such that $f_{W_t}(\bm{x}) \coloneqq \bra{0} W^\dagger_{t}(\bm{x}) Z^{(0)} W_{t}(\bm{x}) \ket{0}$ satisfies
	\begin{align}
		\norm{f(\bm{x}) - f_{W_t}(\bm{x})}^2_{L^2(v)} &= \int_{[0,1]^d} (f(\bm{x}) - f_{W_t}(\bm{x}))^2 \nu(x) \odif{x}\\
		&=\int_{\cup_{\bm\eta}Q_{\bm{\eta}} \bigcup \Lambda(d, K,\Delta)} (f(\bm{x}) - f_{W_t}(\bm{x}))^2 \nu(x) \odif{x}\\
		&\leq (d^{s+\beta/2}K^{-\beta})^2 + 4dK^{1-d}.
	\end{align}
    The width of the PQC is $O(d\log K + \log s + s\log d)$, the depth is $O(s^2 d^sK^d(\log s + s\log d + d\log K)) + \frac{1}{\Delta}\log K)$, and the number of parameters is $O(sd^s(s+d+K^d) + \frac{d}{\Delta} \log K)$.
\end{corollary}

\subsection{Comparison of ``global'' and ``local'' approaches in this work}
We note that we have presented two distinct methodologies for constructing PQC models with UAP properties aimed at approximating continuous functions. In \cref{thm:approx_Lipschitz} and \cref{thm:approx_holder}, we establish PQC models, guided by the multivariate Bernstein polynomials and the Taylor expansion of multivariate continuous functions, respectively. We categorize these approaches as ``local'' and ``global''. We proceed to conduct a comprehensive comparative analysis of these two strategies in the context of approximating Lipschitz continuous functions. For the subsequent analysis, we set $\beta=1$, thus $s=0$ in \cref{thm:approx_holder}, in accordance with the Lipschitz continuous property exhibited by the target function.

The approximation error associated with the global approach can be bounded as ${(2^dd \ell^2)}/{(n\varepsilon^2)}+\eps$. By selecting $n=(2^dd \ell^2)/{\varepsilon^3}$, we ensure an approximation error of $2\varepsilon$. Concurrently, the corresponding number of trainable parameters scale as $O\bigl(2^{d^2}d^{d+1}\ell^{2d}/\varepsilon^{3d}\bigr)$.
In contrast, the local approach exhibits an approximation error scaling as $\sqrt{d}K^{-1}+\varepsilon$. Setting $K=\sqrt{d}/\varepsilon$ ensures a $2\varepsilon$ approximation error, with the number of trainable parameters scaling as $O\left(d^{d/2}/\varepsilon^d\right)$. These findings highlight the advantage of the local approach for approximating continuous functions. More importantly, the approximation error proposed by the local method approaches the optimal convergence rate established in \citet{shen2022optimal}. A formal comparison between PQCs and classical deep neural networks is stated in the next section.

\section{Comparison with related works in classical machine learning}\label{Appendix:E}
\begin{table*}
\vspace{-0.66cm}
\renewcommand\arraystretch{1.2}
\begin{threeparttable}\caption{\label{tab:table11}\small{\textbf{Approximation errors of PQCs and ReLU FNNs}}}
		\setlength{\tabcolsep}{1.5mm}
		\renewcommand{\arraystretch}{1.5}
		\scriptsize{{
				\begin{tabular}{lllllll}
					\toprule
					\textbf{Approach} & \textbf{Target} & \textbf{Width} & \textbf{Depth} & \textbf{Number of parameters} & \textbf{Approximation error} \\
					\midrule
                    PQC & $d$-var.\ deg.-$s$ monomial & $O(d)$ & $O(s)$ & $O(d+s)$ & $0$ \\
				ReLU FNN~\cite{lu2021deep} & $d$-var.\ deg.-$s$ monomial & $O(N+s)$ & $O(s^2M)$ & $O((N^2+s^2)s^2M)$ & $O(sN^{-sM})$ \\ 
					\hline
				Nested PQC & $C_u^s([0,1]^d)$ & $O(d\log K+s\log d)$    & $O(K^dd^s)$ & $O(K^dd^{s+1})$    & $O(d^{2s}K^{-s})$ \\
				$\relu \text{ FNN}^\mathrm{i}$ ~\cite{lu2021deep} & $C_u^s([0,1]^d)$  & $O(s^{d+1}N)$ & $O(s^2M)$  & $O(s^{2d+4}K^{d/2}N)$ & $O(s^d8^sK^{-s})$  \\
					\bottomrule
				\end{tabular}
		}}
		\begin{tablenotes}[flushleft]
			\item [i] Satisfying $NM = \Theta(K^{d/2})$.
		\end{tablenotes}
	\end{threeparttable}
 \vspace{-0.56cm}
\end{table*}

In this subsection, we conduct a comparative exploration of PQCs and classical deep neural networks, focusing on critical aspects, including model size, the number of trainable parameters, and approximation error. To establish a meaningful benchmark, we turn our attention to deep feed-forward neural networks (FNNs) distinguished by the incorporation of rectified linear unit (ReLU) activation functions. FNNs represent the foundational class of neural networks, characterized by a unidirectional flow of information, commencing from the input layer and traversing through one or more hidden layers before culminating at the output layer. This architectural design ensures the absence of cyclic dependencies or loops among nodes within each layer. The ReLU activation function, mathematically defined as $\relu(x):=\max(x, 0)$, has gained prominence across diverse domains, including but not limited to image recognition~\cite{he2016deep, ren2015faster} and natural language processing~\cite{yang2019xlnet,devlin2019bert}. Its popularity in feed-forward networks stems from its efficacy in facilitating the convergence of function approximation during network training. Additionally, a recent study~\cite{zhang2023deep} has affirmed that classical neural networks employing commonly utilized activation functions can be effectively approximated by ReLU-activated networks while maintaining a mild increment in network size. Readers are also referred to some other excellent works related to ReLU
networks~\cite{yarotsky2017ErrorBoundsApproximations,Ingo2020Errorbounds,Johannes2020Nonparametricregression}.

In particular, \citet{shen2022optimal} have proposed the optimal approximation error to approximate any Lipschitz function. \citet{lu2021deep} have provided a nearly optimal approximation error to approximate any smooth function using $\relu$ FNNs. For clarity, the comparison of our results with theirs is summarized in Table.~\ref{tab:table11}. It is pertinent to observe that, in the majority of practical instances, the smoothness coefficient $s$ of the target function tends to be modest since most functions to be approximated is not very smooth. Additionally, within practical scenarios, particularly in domains like image recognition and natural language processing, the dimensionality $d$ of input data is substantially large. Consequently, within this context, we identify terms that solely rely on the variable $s$ as constants and $d\gg s$ within Table~\ref{tab:table11}.

We extend our investigation by quantifying the performance of PQCs and FNNs in terms of the model size and the number of parameters for approximating $s$-smooth functions $C^s_u([0,1]^d)$. Notably, we discover that in cases where the target function adheres to certain norms of smoothness, PQCs exhibit a notable improvement in approximating this function in terms of the model size and the number of parameters.

\noindent
\textbf{Model size.} In particular, we explore the comparison of PQC and FNN model sizes when they yield the same approximation error $\eps$ (say some constant). Here, we use a straightforward measure, the product of width and depth, to gauge the model size. By setting approximation error as $\eps$, the size of PQC and FNN scale as $O(K_Q^dd^{s+1})$ and $O(K_C^{d/2}s^{d+3})$, respectively, where $K_Q=\Theta(d^2/\eps^{1/s})$ and $K_C=\Theta(s^{d/s}/\eps^{1/s})$.

Remarkably, when $2\leq s < d$, an intriguing observation emerges: the ratio of model sizes between PQCs and FNNs~\cite{lu2021deep} exhibits a scaling behavior of $O(\eps^{-d/(2s)}/s^{d^2-d\log_sd})$. Our comprehensive analysis concludes that in situations where the smoothness threshold is satisfied, PQCs boast a significantly smaller model size compared to FNNs. 

\noindent
\textbf{Number of trainable parameters.} 
In the present investigation, we delve into the comparative analysis of the number of trainable parameters of PQC and FNN under the premise of yielding comparable approximation errors. From the perspective of approximation theory, the count of parameters serves as a standard metric for assessing model degrees of freedom and expressing model expressiveness. By setting approximation error as $\eps$, the number of trainable parameters of PQC and FNN scale as $O(K_Q^dd^{s+1})$ and $O(K_C^{(1+\lambda_0)d/2}s^{2d+4})$, respectively. Here, the hyperparameter $\lambda_0\in (0, 1)$ signifies FNN's width. 

Remarkably, through our analysis, we have uncovered that when $2\leq s< d$, the relationship between the number of trainable parameters of PQCs and FNNs~\cite{lu2021deep} demonstrates a scaling pattern characterized by $O(\eps^{-(1-\lambda_0)d/(2s)}/s^{(1+\lambda_0)d^2-d\log_sd})$. As a consequence, the number of trainable parameters of PQCs significantly reduces compared to that of FNNs. 

\noindent
\textbf{Approximating monomial.} 
Here, we conduct a comparative performance analysis of PQC and FNN in approximating monomial functions of degree $s$. Within this specialized target function space, PQCs exhibit distinct advantages in terms of width, depth, model size, and the number of trainable parameters. Notably, PQCs possess the unique capability to capture the dynamics of monomial functions precisely, eliminating the need for approximation and thereby offering a compelling advantage. These advantages position PQCs as promising candidates for outperforming FNNs when addressing more complex target function spaces.

\end{document}